\documentclass[journal,onecolumn,11pt]{IEEEtran}
\usepackage[doublespacing,nodisplayskipstretch]{setspace}
\usepackage[utf8]{inputenc} 
\usepackage[T1]{fontenc}
\usepackage[colorlinks,pdfstartview=FitH,citecolor=blue,linkcolor=blue,urlcolor=blue]{hyperref}
\usepackage{url,comment}
\usepackage{ifthen}
\usepackage{cite}
\usepackage{graphicx}
\usepackage{amsmath,amssymb,amsthm}
\usepackage{cite,soul}
\usepackage{tikz}
\usepackage[capitalize]{cleveref}
\usepackage{todonotes}
\usepackage{array}
\usepackage[flushleft]{threeparttable}
\usepackage[inline]{trackchanges}
\addeditor{Swanand} 

\newcommand{\revise}[1]{{\color{blue}#1}}

\newcommand{\GJ}[1]{\todo[color=orange!25, inline]{ Gauri: #1}\index{Gauri: !#1}}

\newtheorem{theorem}{Theorem}
\newtheorem{lemma}{Lemma}

\newtheorem{definition}{Definition}
\newtheorem{example}{Example}
\newtheorem{proposition}{Proposition}
\newtheorem{remark}{Remark}
\newcommand{\prob}[1]{{\mathbb P}}
\newcommand{\cobject}{c}
\ifCLASSINFOpdf
\else
\fi

\begin{document}
%
\title{Service Rate Region: A New Aspect of Coded\\[-1ex] Distributed System Design}
%
%
%

\author{Mehmet Akta\c{s},  Gauri Joshi,~\IEEEmembership{Member,~IEEE,} Swanand Kadhe,~\IEEEmembership{Member,~IEEE,} \\Fatemeh Kazemi,~\IEEEmembership{Student Member,~IEEE,} Emina Soljanin,~\IEEEmembership{Fellow,~IEEE}
\thanks{
M.\ Akta\c{s} is with The MathWorks Inc., Natick MA 01760-2098, USA, email: mfatihaktas@gmail.com,
G.\ Joshi is with the Department of Electrical and Computer Engineering, Carnegie Mellon University, Pittsburgh PA 15213, USA, email: gaurij@cmu.edu,
S.~Kadhe is with the Department of Electrical Engineering and Computer Sciences, UC Berkeley, Berkeley CA 94720, USA, email: swnanand.kadhe@berkeley.edu,
F.\ Kazemi is with the Department of Electrical and Computer Engineering, Texas A\&M University, College Station, TX 77843, USA, email: fatemeh.kazemi@tamu.edu,
E.~Soljanin is with the Department of Electrical and Computer Engineering, Rutgers, The State University of New Jersey, Piscataway, NJ 08854, USA, e-mail: (see https://www.ece.rutgers.edu/emina-soljanin).}
\thanks{Some parts of this paper, in particular, parts of \Cref{sec:MDS,sec:simplex_comb_opt} appeared in the Proc.\ of the 2017 Allerton conference  \cite{aktas2017service}}}

%
%

\markboth{IEEE Transactions on Information Theory, 
June~2021} {}
%


\maketitle

\begin{abstract}

Erasure coding has been recognized as a powerful method to mitigate delays due to slow or straggling nodes in distributed systems. This work shows that erasure coding of data objects can flexibly handle skews in the request rates. Coding can help boost the \emph{service rate region}, that is, increase the overall volume of data access requests that the system can handle. This paper aims to postulate the service rate region as an important consideration in the design of erasure-coded distributed systems. We highlight several open problems that can be grouped into two broad threads:
1) characterizing the service rate region of a given code and finding the optimal request allocation, and
2) designing the underlying erasure code for a given service rate region. As contributions along the first thread, we characterize the rate regions of maximum-distance-separable, locally repairable, and Simplex codes. We show the effectiveness of hybrid codes that combine replication and erasure coding in terms of code design. We also discover fundamental connections between multi-set batch codes and the problem of maximizing the service rate region.
\end{abstract}

%
\IEEEpeerreviewmaketitle

\section{Introduction}
\label{sec:intro}




The emergence of flexible and affordable cloud storage and computing has resulted in an exponential growth in the amount of data that is stored and processed in cloud data centers. This increase in data is accompanied by a similar rapid increase in the volume of users accessing it, resulting in frequent contention for shared cloud resources. A simple way to handle more requests in a fast and reliable fashion is to replicate data at multiple nodes \cite{raid_1988, shvachko2010hadoop}. However, replication can be expensive in terms of storage, especially when the data is updated frequently. Moreover, the popularity of different data objects 
can vary drastically across objects and over time. While edge caches can handle skews in popularity by selectively increasing the number of replicas of the `hot' or popular objects \cite{Edge:YadgarKAS19,maddahali2016coding, Caching:BreslauLC99,Caching:RabinovichMS02,shanmugam2013femtocaching,hamidouche2014many}, such quick adaptation may not be possible in the data-center setting, especially for large data objects that are used in data analytics or machine learning 
applications. Besides, in caching essentially the limited capacity of the backhaul link is considered as the main bottleneck of the system, and the goal is usually to minimize the backhaul traffic or maximize the cache hit rate by prefetching the popular contents at the edge nodes of limited storage capacity. However, caching does not aim to handle the scenarios such as live streaming where many users want to get the same content simultaneously given the limited service capacity (bandwidth) of the nodes in the network. 

In this work, we propose the use of erasure coding to handle data access requests to distributed storage systems. We consider coded distributed systems where $k$ different data objects (rather than $k$ chunks of one object) are erasure coded into $n$ coded objects which are stored on $n$ nodes. We consider \emph{heterogeneous} requests to access these objects at rates $\lambda_1$, $\lambda_2$, \dots , $\lambda_k$, respectively. Each of the $n$ nodes can serve at most $\mu$ rate of requests. Thus, the total request rate allocated to each node must not exceed $\mu$. Under these constraints, we aim to characterize the set of achievable vectors ($\lambda_1$, $\lambda_2$, \dots , $\lambda_k$), which we refer to as the \emph{service rate region} of a coded distributed system. Since the nodes storing coded objects can be used to partially serve requests for any of the objects included in that coded combination, coded distributed systems are more flexible and can have a different (possibly more favorable) service rate region than an uncoded system with the same number of nodes. We illustrate this through the motivating example below.

\vspace{0.1cm} \noindent \textbf{Motivating Example.} Consider an example shown in~\Cref{fig:storage_eg}, where two objects $a$ and $b$ are redundantly stored on $4$ nodes. \Cref{fig:storage_eg} (left) shows $3$ redundant storage schemes: replication, coding, and replication and coding combined. Given that each node can serve $\mu=1$ request per second, we want to maximize $\lambda_a$ and $\lambda_b$, the rate of requests for $a$ and $b$ that can be supported. Object $a$ can be downloaded from the node storing $a$, or from two nodes that store coded combinations of $a$ and $b$.

\begin{figure}[t]
\begin{center}
\begin{tikzpicture}
\node at (0,0) {\includegraphics[scale=0.99]{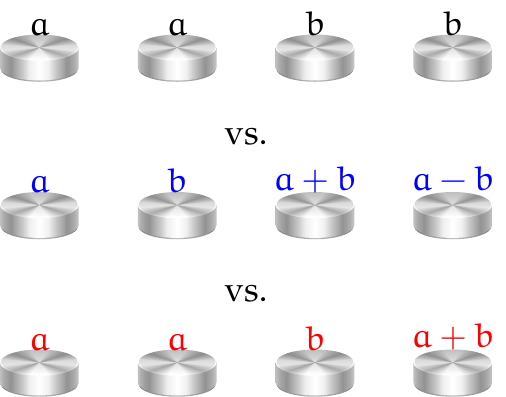}};
\node at (8,0) {\includegraphics[scale=0.99]{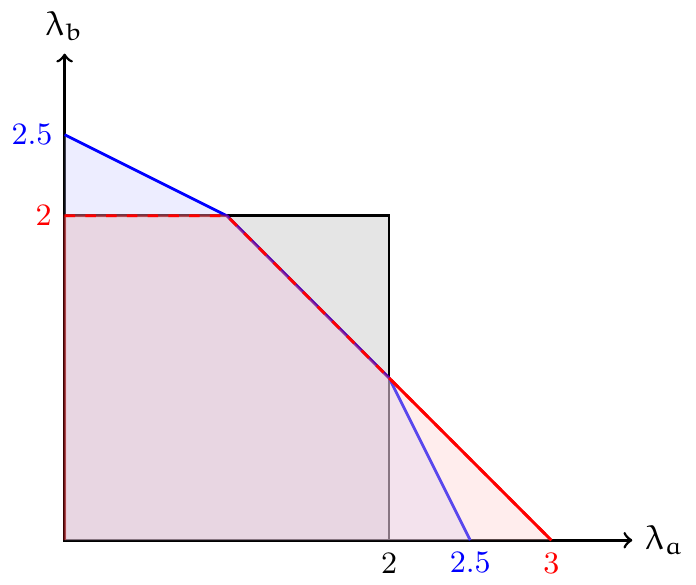}};
\end{tikzpicture}
\end{center}
\caption{(left) Replicated, coded, and hybrid systems with $n=4$ nodes storing $k=2$ files. (right) Service rate regions of the three systems when the service capacity of each node is $\mu=1$. The regions have the same areas. Coding can handle the skews in request arrival rates $\lambda_a$ and $\lambda_b$ for the two stored objects.
\label{fig:storage_eg}}
\end{figure}

\Cref{fig:storage_eg} (right) shows the service rate regions of the $3$ storage systems. The replicated system can achieve the square service rate region with $0 \leq \lambda_a, \lambda_b \leq 2$; this is because there are two copies of each object and each node can support $\mu = 1$ rate of requests. The coded system with two nodes storing $a+b$ and $a-b$ respectively instead of uncoded copies of $a$ and $b$ achieves the blue colored shaded service rate region. This system can handle skews in $\lambda_a$ and $\lambda_b$ better than the replicated system when one of the two objects $a$ and $b$ are more frequently accessed but both objects are unlikely to be popular simultaneously. The service rate region of a combined replication and coding system (shown in red) can better support asymmetries in the demands $\lambda_a$ and $\lambda_b$ and is the best choice when the request rate for $a$ is expected to be larger than that of $b$.


\vspace{0.1cm} \noindent \textbf{Related Previous Work.} Erasure codes are often used in distributed storage systems to improve reliability against disk failures \cite{raid_1988}. A class of codes, which are commonly used in distributed storage, are systematic maximum-distance-separable (MDS) codes \cite{berlekamp1968algebraic}, where an object is divided into $k$ chunks (also called stripes) that are then encoded into $n$ chunks by adding $n-k$ redundant parity-check chunks, thus providing resilience to the failure of up to $n-k$ nodes.
%
Until recently, erasure coding in storage systems was mostly used for `cold' or less frequently accessed and less latency-sensitive data. This is because, erasure coded systems require access to $k$ nodes (each storing one of the $k$ chunks of an object) in order to download the object, and slowdown of any one of these nodes can become a bottleneck in serving data access requests. Thus, replication is generally preferred over erasure coding for hot and latency-sensitive data access. Recently, the idea of redundant data access requests, that is, sending requests to all $n$ nodes of an erasure-coded system and waiting for any $k$ nodes to respond, has been shown to be effective in overcoming such tail latency due to straggling nodes \cite{joshi2015efficient, joshi2017efficient, joshi2012coding, shah2016when, rashmi2016ec}. Similar redundancy ideas are also used in the context of distributed computing \cite{raaijmakers2020achievable, gardner2015reducing, gardner2016power, lee2017speeding, dutta2016short, mallick2020rateless}. However, most of these works on faster data access from coded distributed systems focus on homogeneous reads, where all $k$ chunks of an object are accessed at the same time. 

Heterogeneous data access has been previously considered in the context of hot data download (see e.g., \cite{KSS:15,DownloadTimeOfAvailabilityCodes:kadheKSS2015, SimplexQueues:AktasNS17,DownloadTimeOfAvailabilityCodes:AktasKS20}) and load balancing (see e.g., \cite{BatchCodesAndTheirApps:IshaiKO04}). In the hot data download context, heterogeneity arises when one of the data objects is highly popular. Works such as \cite{KSS:15,DownloadTimeOfAvailabilityCodes:kadheKSS2015, SimplexQueues:AktasNS17,DownloadTimeOfAvailabilityCodes:AktasKS20} analyze the expected latency experienced by requests that are replicated across the hot object's recovery sets. Previous works on load balancing heterogeneous requests to coded systems that arises when a batch of simultaneous requests consists of different numbers of requests for each of the $k$ objects. Special coding schemes, known as multi-set batch codes \cite{BatchCodesAndTheirApps:IshaiKO04}, have been proposed to allow serving such requests with a balanced amount of downloaded data across the servers (see e.g.\ \cite{SilbersteinG16, RawatSDG16} and references therein). In this paper, we do not impose any limits on data access heterogeneity, that is, on the arrival rates $\lambda_1$, $\lambda_2$, \dots , $\lambda_k$. Our main focus is the service rate region, that is, the set of request arrival rates that the system can support; we discuss connections to download latency and load balancing in Sec.~\ref{sec:open_problems}.

The main difference between this and, nearly all, recent work on coded distributed storage is that the proposed work primarily addresses the external uncertainty in the storage systems (download requests fluctuations) rather than the internal uncertainty (e.g., straggling) in operations of the system itself.
A related line of work by some of the authors (addressing external uncertainty) considers systems with uncertainty in the mode and level of access to the system \cite{service:Noori2016SM, service:PengS18, allocations:PengNS21}. 

\vspace{0.1cm} \noindent \textbf{Goals and Organization of this paper.}  The main goal of this paper is to propose the service rate region as an important paradigm in the design of erasure coded distributed systems. Characterizing service rate region gives us a clear picture of the collective rate of requests that can be supported by the system as well as its robustness to heterogeneous request patterns where some objects are more frequently accessed than others. In this paper we highlight two main threads of ongoing and future research directions that explore different aspects of the service rate region of coded distributed systems: 1) designing optimal policies to split incoming requests across the nodes in order to maximize the achievable service rate region for a given storage scheme, and 2) designing the underlying code to maximize the service rate region or to cover a given region with minimum storage, as introduced in \Cref{sec:two_problems}. 

The first problem of optimal request splitting can be formulated as a constrained optimization problem. However, it cannot be trivially solved using linear solvers because the number of optimization variables is large and the problem becomes computationally intractable. As contributions along this thread, we characterize the rate region of some well-known classes of codes such as maximum-distance-separable (MDS), locally recoverable (LRC), Simplex (also called Hadamard) codes, and first-order Reed-Muller (RM) codes. These analyses provide insights into how the service rate region is affected by the length and the rate of the underlying code. We can provably find the best service rate region for certain classes of codes such as MDS codes and Simplex codes, but finding the optimal request splitting scheme for other code classes still has many open questions. We highlight three different techniques to solve the problem of optimal request splitting to maximize the service rate region: 1) \Cref{sec:MDS} uses a waterfilling algorithm to find the rate region of MDS and LRC codes, 2) \Cref{sec:simplex_comb_opt} uses fractional matching and vertex cover on graph representation of codes (which we introduce in \Cref{sec:graph_representations}) to find the rate region of Simplex codes, and 3) the geometric approach used in \Cref{sec:reed_muller} to find the rate region of first-order RM codes. Along the second thread of designing the underlying code, in \Cref{sec:two_problems} we highlight the complementary problems of maximizing the rate region of a given number of servers and covering a desired rate region using a minimum number of nodes. The key insight from this exploration is that hybrid codes that carefully combine replication and erasure coding of data (such as the ($a$, $a$, $b$, $a+b$) system considered in the motivating example above) are best-suited for maximizing the service rate region in many cases. 



A crucial goal of this paper is to provide a comprehensive list of open problems in connection with this emerging idea of using the service rate region to guide the design and analysis of erasure coded systems. These problems are of interest to the information and coding theory, (combinatorial) optimization, as well as the queueing/networking communities and can bridge interdisciplinary connections between them.  In \Cref{sec:open_problems} we discuss specific problems such as service rate region considerations in the design of codes, designing codes that cover a given request rate distribution, latency analysis of erasure coded systems, service rate region with redundant requests. In \Cref{sec:batch_codes_connection} we discover fundamental connections between batch codes and the problem of maximizing the service rate region. In fact, codes that maximize the service rate region are a generalization of primitive multi-set batch codes where the demands $\lambda_a$ and $\lambda_b$ of different objects are not constrained to be integers.

Several problems presented in this paper not only require expertise from different areas, but have also already been addressed in those areas in some special forms and under different names. We will explain how some problems associated with the service rate region generalize some previously studied problems. We will also present several problems that belong to but have not been asked yet in certain areas, and thus experts in those areas could potentially provide answers with not too much  difficulty. 

\vspace{0.1cm} \noindent \textbf{How to read this paper.}\space Depending on the reader’s main interest and prior knowledge, different sections of paper may or may not be relevant.
\Cref{sec:SM,sec:two_problems,sec:storage_schemes} should be read by everyone since they provide the system model, problem formulation, some preliminary notions as well as several examples running throughout the paper. Information theorists and queueing theorists may be primarily interested in \Cref{sec:MDS}, theoretical computer scientist in \Cref{sec:simplex_comb_opt}, and coding theorists in  \Cref{sec:reed_muller}. Each of the \Cref{sec:MDS,sec:simplex_comb_opt,sec:reed_muller} can be read immediately after the introductory \Cref{sec:SM,sec:two_problems,sec:storage_schemes}, and independently of the other two. Similarly, the readers can select open problems from the large list in \Cref{sec:open_problems} according to their interests and expertise -- this section includes performance analysis and networking problems as well as coding theory and data allocation problems. 



\section{Distributed Service Model\label{sec:SM}}
We distinguish between two functional components at each node: 
one for data storage and the other for service request processing. This is indicated in~\Cref{fig:SimplexStore} illustrating a system of $n=7$ servers.
\begin{figure}[hbt]
    \centering
    \includegraphics{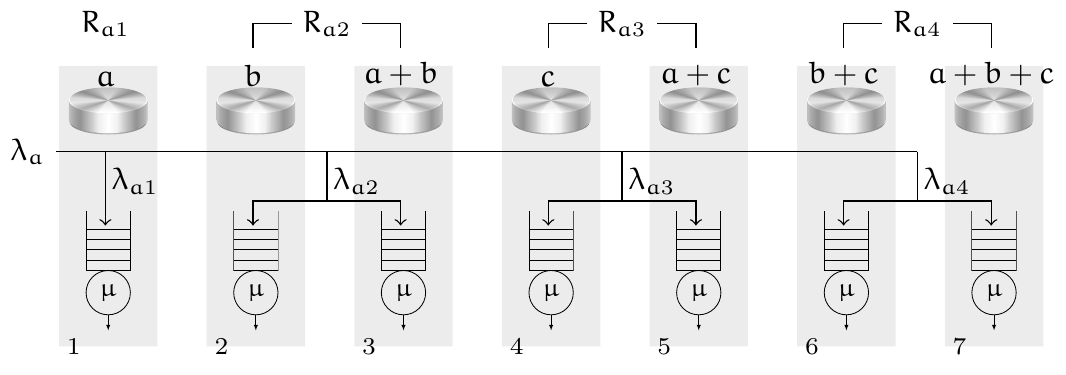}
    \caption{A distributed system storing $k=3$ data objects $a, b, c$ over $n=7$ nodes. Each node stores a symbol of the $[7,3]$ Simplex code. 
    $R_{a_i}$'s represent the recovery sets for $a$, and $\lambda_{a_i}$ denotes the portion of requests rate for serving $a$ that is assigned to $R_{a_i}$  such that $\lambda_{a}=\lambda_{a_1}+\lambda_{a_2}+\lambda_{a_3}+\lambda_{a_4}$ holds. 
    }
    \label{fig:SimplexStore}
\end{figure}

\subsection{Data Storage Model}\label{sec:storage_Model}
Consider that we have $k$ data objects (to be) redundantly stored across $n\ge k$ servers. We assume all data objects are of the same size, and all servers have a storage capacity of one object. Mathematically objects are represented as elements of some finite field $\mathbb{F}_q$. Each server can store a linear combination of data objects, which amounts to a coded object of the same size. We assume that the same erasure code is used for all the objects in the system. 
Simple replication of objects, that is, storing identical copies, is allowed and, as we will see later, often to a certain extent desirable. It is worth to note here that the assumptions we make in our storage model are common in the prior work, see, e.g.,~\cite{joshi2012coding,shah2016when,SimplexQueues:AktasNS17,DownloadTimeOfAvailabilityCodes:AktasKS20}. We denote the coded objects as $\cobject_1, \cobject_2, \ldots, \cobject_n$. Because of redundancy, any data object can be recovered (computed) from multiple sets of encoded objects. 

\begin{definition}
A \ul{recovery set} for a coded object $c_i\in\mathbb{F}_{q}$ is a minimal set of coded objects $R$ such that there exists a recovery function $rec:\mathbb{F}_q^{|R|}\to\mathbb{F}_q$ satisfying $rec(R) = c_i$.
\end{definition}

Each object has at least one recovery set (the object itself), and may potentially have multiple recovery sets. We denote by $R_{i,1},\dots,R_{i,t_i}$ the $t_i$ recovery sets of object $i$. \Cref{fig:SimplexStore} shows a system where three data objects $a$, $b$ and $c$ are encoded by a $[7,3]$ binary Simplex code. The recovery sets for object $a$ are its systematic copy $(a)$ and the pairs of linear combinations $(b, a+b)$, $(c , a+c)$ and $(b+c, a+b+c)$. For a systematic MDS code, the recovery sets for each object are its uncoded copy and any size-$k$ subset of the remaining $n-1$ nodes, and thus the number of recovery sets for each data object is $1+\binom{n-1}{k}$. For instance for the $[4,2]$ MDS coded system shown in blue in \Cref{fig:storage_eg}, the recovery sets for object $a$ are $(a)$, $(a+b, a-b)$, $(b, a+b)$ and $(b, a-b)$.

\subsection{Data Access Model}



We describe two data access models which are two different ways of implementing resource sharing among the incoming data access requests. We refer to them as the queuing and the bandwidth model. Both of these access models result in similar mathematical formulations of the service rate region, which we define in \Cref{sec:two_problems}.
In both models, we assume that requests to download object $i$ arrive at rate $\lambda_i$, and that the service rate at each server is $\mu$ requests per unit time.

\subsubsection{Queueing Model}
\label{sec:queuing-model}
Requests sent to each server are placed in a queue at the server (the queue can follow either first-come-first-served or any other scheduling discipline). In order to maintain the stability of the queue at each server, the total request arrival rate at each server should not exceed its service rate $\mu$. Our goal is to characterize the service rate region, that is, the set of arrival rates $(\lambda_1, \lambda_2,\dots, \lambda_k)$ for the $k$ objects that can be supported by the system.

\subsubsection{Bandwidth Model}
\label{sec:bandwidth-model}
Suppose that storage drives associated with each node can concurrently serve only a limited number of data access requests. This is because each drive has an I/O bus with a finite access bandwidth $W$ bits/second, and a download request requires streaming at a fixed bandwidth of $b$ bits/second. Therefore, a node can serve only $\mu=W/b$ number of requests concurrently. Let $\lambda_i$ be the number of requests for file $i$ that are simultaneously present in the system. 
Our goal is to characterize all request combinations for the $k$ data objects $(\lambda_1,\lambda_2,\dots,\lambda_k)$ that can simultaneously be served by the system. Note that unlike the queuing model, $\lambda_i$'s are integers in this case. In \Cref{sec:batch_codes_connection} we define the notion of the integral service rate region for this model and show how it is fundamentally connected to batch codes \cite{Ishai:04:batch-codes}. 

\subsection{Extension to Coded Computing}
Although we present the queueing and bandwidth models in the context of data-access, these models and the resulting formulation of the service rate region can be directly applied to determine the service rate region of coded computing systems as well. For example, consider a system of $n=4$ servers which store large matrices $\mathbf{A}$, $\mathbf{B}$, $\mathbf{A}+\mathbf{B}$ and $\mathbf{A}+2\mathbf{B}$ respectively. Each request in the rate $\lambda_A$ has a query vector $\mathbf{x}$ and its goal is to obtain a matrix-vector product $\mathbf{A}\mathbf{x}$, where as the requests in $\lambda_B$ seek to compute $\mathbf{B}\mathbf{x}$. The task of computing $\mathbf{A}\mathbf{x}$ can be completed either by sending $\mathbf{x}$ to the node storing $\mathbf{A}$, or by sending it to any two out of the three remaining servers. In the queueing model, requests are placed in a queue at each server with service rate $\mu$, whereas in the bandwidth model, $\mu$ is the maximum number of computation tasks that each server can handle simultaneously.

\section{Service Rate Regions: Two Problems of Interest}
\label{sec:two_problems}

In this section we describe two classes of problems that arise in the context of using the service rate region as a metric to design erasure coded distributed systems: 1) maximizing the service rate region of a given storage scheme via optimal resource allocation, and 2) designing a storage-efficient erasure coded scheme to cover or achieve a desired service rate region. The purpose of this section is not only to formulate these problems, but also to highlight the variety in the mathematical techniques that are applicable to service rate region problems. These techniques bridge deep connections between fundamental coding theory and resource allocation problems. In subsequent sections we mainly address the first problem described above, of maximizing the service rate region of a given storage allocation. 

\subsection{Finding the Service Rate Region of a Given Storage Scheme}
\label{sec:charac_rate_region}

Given distributed system with $k$ data objects stored on $n$ servers, we first present the problem of maximizing the service rate region by optimally splitting incoming request rates $\boldsymbol{\lambda} = (\lambda_1, \lambda_2, \dots, \lambda_k)$ across the $n$ servers. The problem of optimally allocating incoming data access requests to one or more servers can be formulated as a linear optimization problem, which we describe below.

Recall that each server has the service rate $\mu$, that is, it can serve $\mu$ requests in unit time. Let us use $\lambda_{i,j}$ to denote the portion of requests for object $i$ that is assigned to the recovery set $R_{i,j}$, $j=1,\dots,t_i$.\footnote{Note that we do not make any assumptions on the arrival process such as Poisson arrivals.} Without loss of generality, suppose that the first $k-1$ elements of the demand vector $\boldsymbol{\lambda}$ are given and we aim to maximize $\lambda_k$. Then the optimal request rate split $\lambda_{i,j}$ for all $i = 1, \dots, k$ and $j = 1, \dots, t_i$ is the solution to the following linear optimization problem:

\begin{align}
& \max_{\lambda_{i,j}: i \in [1, k], j \in [1, t_i]} \lambda_k \qquad\text{s.t.}\qquad \label{eq:opt_problem}\\
& \sum_{j=1}^{t_i}\lambda_{i,j}=\lambda_i~\text{for}~  1\le i\le k, \label{eq:allocation-constraints_1}\\
& \sum_{i=1}^k \sum_{\genfrac{}{}{0pt}{}{1\le j\le t_i}{\ell\in R_{i,j}}}\lambda_{i,j} \le \mu \text{ for } 1 \leq l \leq n, \label{eq:allocation-constraints_2}\\
& \lambda_{i,j} \geq 0, \text{ for }  1\leq i\leq k, 1\leq j\leq t_i.
\label{eq:allocation-constraints_3}
\end{align}





The first set of constraints \eqref{eq:allocation-constraints_1} guarantees that the demands for all objects are served. The second set of constraints \eqref{eq:allocation-constraints_2} ensures that the total demand assigned to each server are within its service capacity limit. Note that these recovery groups $R_{i,j}$ can overlap.
The set of constraints \eqref{eq:allocation-constraints_2} ensures the stability of the system, where the request arrival rates must be below the corresponding service rates. The rates at which the system is able to serve requests should not be confused with information rates that are used to provide the service. For example, if satisfying a request for object $a$ involves downloading objects $b$ and $a+b$, then the user requesting $a$ will also get $b$. He will receive twice as much information as requested but not more service as he did not request $b$.

\begin{definition}
Given a distributed system with $k$ data objects stored on $n$ servers, a set $\{\lambda_{i,j}:1\leq i\leq k, 1\leq j\leq t_i\}$ satisfying~\eqref{eq:allocation-constraints_1}-\eqref{eq:allocation-constraints_3} is referred to as a \ul{valid allocation}.
\end{definition}

\begin{definition}
Given a distributed system with $k$ data objects stored on $n$ servers, the set of all achievable demand vectors $\boldsymbol{\lambda} = (\lambda_1, \lambda_2,\dots, \lambda_k)$ is referred to as the \ul{service rate region} of the system. 
\end{definition}


Note that the queueing and bandwidth data access models, despite their different practical meanings, give rise to the service rate region which is the solution to the optimization problem described by \eqref{eq:opt_problem}-\eqref{eq:allocation-constraints_1} above.


There are scenarios wherein each user occupies the entire bandwidth of the server they are accessing. This can happen, for instance, when users are streaming from low-bandwidth edge devices and each user needs to be served at a specific rate. We call a service rate region under this constraint as an \textit{integral service rate region} (formally defined in~\Cref{sec:batch_codes_connection}). 



Depending on the number of objects and nodes, and the coding scheme, this problem can be very computationally expensive to solve. For example, for systematic MDS codes the number of recovery sets is $\binom{n-1}{k} + 1$, which grows exponentially as $n$ and $k$ increase. 
Below we highlight three varied approaches that allow us to solve this problem in closed-form for certain classes of codes: 1) water-filling algorithms similar to those used for proving capacity theorems in information theory, 2) combinatorial optimization on graphs, and 3) a geometric approach. Later in the paper we use these approaches to characterize the service rate regions of MDS and locally recoverable codes (\Cref{sec:MDS}), Simplex codes (\Cref{sec:simplex_comb_opt}),  and Reed-Muller codes (\Cref{sec:reed_muller}), respectively. 
We summarize the coding schemes considered in this paper and the approaches used to characterize their service rate regions in~\Cref{tab:summary-of-codes}.
\subsubsection{Water-filling Algorithm} The water-filling or water-pouring algorithm is a common technique to allocate power across multiple channels in a digital communication system \cite{gallagher1968information, cover2006elements} . 
It treats the channels as vessels with uneven bottom levels, proportional to their noise variance. Power is first allocated to the least noisy channel until its signal plus noise reaches the level of the next lowest noise level. In \Cref{sec:MDS}, we extend the concept of water-filling to allocate requests to servers by treating each server as a vessel with capacity $\mu$, when the storage scheme is an MDS or a locally recoverable code. Each small volume $\epsilon > 0$ of requests within the total $\sum_{i=1}^{k} \lambda_i$ demand is assigned to a recovery group by adding $\epsilon$ volume of water to the corresponding vessels. We show that allocating each request to the least-loaded nodes in the smallest recovery group (first the systematic nodes, then the local parities and then the global parities) maximizes the service rate region. In other words, water-filling resource allocation achieves the optimal system throughput. See \Cref{sec:MDS} for a formal description of this technique and bounds on the resulting service rate region. An ongoing research direction is to explore the use of water-filling for other classes of codes and proving its throughput-optimality. 

\subsubsection{Combinatorial Approach}
\label{sec:combinatorial-approach-summary}
This approach establishes a significant connection between the service rate problem and the well-known fractional matching problem in (hyper)graphs. A connection between distributed storage allocation problems (see~\cite{sardari2010memory,allocation:LeongDH12} and references therein) and matching problems in hyper-graphs has been observed in computer science literature~\cite{Matching:AlonFH12} (see also~\cite{Matching:KaoDLH13}). In particular, it was noted that the uniform model of distributed storage allocation considered in~\cite{sardari2010memory} leads to a question which is asymptotically equivalent to the fractional version of a long-standing conjecture by Erd\H{o}s~\cite{erdHos1965problem} on the maximum number of edges in a uniform hypergraph.

Here, we introduce a novel technique for constructing a special graph representation of a linear code. In particular using this approach, the following results are shown: 1) equivalence between the service rate problem and the well-known fractional matching problem and 2) equivalence between the integral service rate problem and the matching problem. These equivalence results allow one to use techniques in the rich literature of the graph theory for solving the service rate problem. Leveraging these equivalence results, it is shown that the maximum sum rates that can be simultaneously served by the system equals the fractional matching number in the graph representation of the code, and thus is lower bounded and upper bounded by the matching number and the vertex cover number, respectively. This is of great interest because if the graph representation of a code is bipartite, then the derived upper bound and lower bound are equal which allows one to establish the maximum sum rates that can be served by the system. Utilizing this result, the service rate region of the binary Simplex codes is characterized whose graph representation is bipartite as shown in Sec.\ref{sec:sim}.


We also show in Sec.~\ref{sec:batch_codes_connection} that the notion of integral service rate region opens up interesting connections with batch codes, a class of codes designed for simultaneous access~\cite{Ishai:04:batch-codes}. Specifically, we show that the service rate problem can be viewed as a generalization of the batch code problem, and the multiset primitive batch codes problem is a special case of the service rate problem when the portion of requests assigned to the recovery sets is limited to be integral.

\subsubsection{Geometric Approach}

Finding the service rate region of a given storage scheme is an optimization problem. One natural way to look at this problem is through the geometric approach, introduced in~\cite{ServiceGeometric:KazemiKS20}, that provides a set of half-spaces whose intersection surrounds the service rate region of a given linear storage scheme. In other words, the geometric approach provides upper bounds (half-spaces) on the sum of each subset of arrival rates in any demand vector $(\lambda_1, \cdots, \lambda_k)$ in the service rate region of a linear code in a more straightforward manner in comparison to other approaches. This technique is of great significance since it allows one to derive upper bounds on the service rates of linear codes without explicitly knowing the list of all possible recovery sets while waterfilling and combinatorial approaches rely on enumeration of all possible recovery sets which becomes increasingly complex when the number of objects $k$ increases.

Using the geometric technique, upper bounds on the service rates of the binary first order Reed-Muller codes and binary Simplex codes are derived. It is worth mentioning that only the cardinality of the recovery sets matters in deriving upper bounds on the service rate of the first order Reed-Muller codes using the geometric approach. Subsequently, it is shown that how the derived upper bounds can be achieved. Furthermore, it is illustrated that given the service rate region of a code, a lower bound on the minimum distance of the code can be obtained. This approach will be discussed further in Sec.~\ref{sec:reed_muller}. For the original observation and more details, see~\cite{ServiceGeometric:KazemiKS20}.

\begin{table}[t]
    \centering
    \caption{Summary of the coding schemes and the techniques used to characterize their service rate regions.}
    \small
    \label{tab:summary-of-codes}
       \setlength{\extrarowheight}{0.5ex}
    \begin{tabular}{l|l|l|p{0.4\linewidth}}  
    Technique & Codes & Results & Description\\
    \hline
         Waterfilling &
         \begin{tabular}{@{}c@{}}\\ Systematic MDS \end{tabular} &
         \Cref{thm:converse_bnd,thm:waterfilling_opt} & 
         \mbox{Characterizes the service rate region for $n-k\geq k$}\\
    \cline{3-4}
         {} & 
         {} &
         \Cref{lem:waterfilling-high-rate-optimality-1,lem:waterfilling-high-rate-optimality-2} &  
         \mbox{Show optimality of waterfilling for $n-k < k$}\\
    \hline
        \begin{tabular}{@{}c@{}}Combinatorial \\ \end{tabular} &
        \begin{tabular}{@{}c@{}}Binary Simplex \end{tabular} &
        \Cref{thm:simplex-codes} &
        \mbox{Characterizes the service rate region}\\
    \cline{2-4}
        {} &
        \begin{tabular}{@{}c@{}} Non-Systematic MDS \end{tabular} & 
        \Cref{prop:rate-region-MDS-using-graphs} &
        \mbox{Characterizes the service rate region}\\
    \cline{2-4}
        {} &
        \begin{tabular}{@{}c@{}}Primitive Multiset Batch \end{tabular} &
        \Cref{prop:batch-codes} &
        \mbox{Shows a relation between batch codes and integral} \mbox{service rate region (cf.~\Cref{def:integral-rate-region})}\\
    \hline
        \begin{tabular}{@{}c@{}}Geometric$^{*}$ \\\end{tabular} &
        \begin{tabular}{@{}l@{}}Binary Simplex\end{tabular} & \Cref{sec:geometric-simplex-codes} & \mbox{Characterizes the service rate region for binary} \mbox{$[7,3]$ Simplex code}\\
    \cline{2-4}
        {} &
        \begin{tabular}{@{}c@{}}Binary First-Order Reed-Muller \end{tabular} & \Cref{sec:geometric-RM-codes} & \mbox{Characterizes the service rate region for binary} \mbox{non-systematic $[8,4]$ Reed-Muller code}\\
    \hline
    \multicolumn{4}{l}{$^{*}$Two illustrative examples are included; the rate regions are characterized in~\cite{ServiceGeometric:KazemiKS20}.}
    \end{tabular}
\end{table}

\subsection{Designing Storage Schemes to Maximize or Cover the Service Rate Region}
\label{sec:code_design_max_rate_region}

Complementary to the problem of finding the service rate region of a given storage scheme, we now discuss the problem of designing the underlying storage scheme to achieve a target service rate region with the minimum number of nodes. The target service rate region represents a known probability distribution of demand vectors $\boldsymbol{\lambda} = (\lambda_1, \lambda_2, \dots, \lambda_k)$ that can be supported by the system. A system designer aims to support this demand distribution using the minimum number of servers or to maximize the volume of the service rate region for a given number of servers.  Below we discuss these two facets of the storage scheme design problem and provide some initial solution perspectives. These problem is largely open and requires fundamental coding theoretic innovations.


\subsubsection{Maximize Service Rate Region with a Given Number of Servers}
Consider a practical scenario where a fixed number of nodes $n$ is available to store $k$ objects, and a coding scheme is to be designed to maximize the service rate region. This problem was considered for $k=2$ in \cite{ServiceCapacity:AktasAJ17} for $k=3$ in \cite{anderson2018service}. For example, consider four different schemes to store $k=2$ files on a system of $n=8$ servers as shown in \Cref{fig:MaxRegion}, where $\alpha$ is a primitive element of $\mathbb{F}_9$ or a larger finite field. 
\begin{figure}[htb]
    \begin{tikzpicture}
    \node at (0,0) {\includegraphics[scale=0.825]{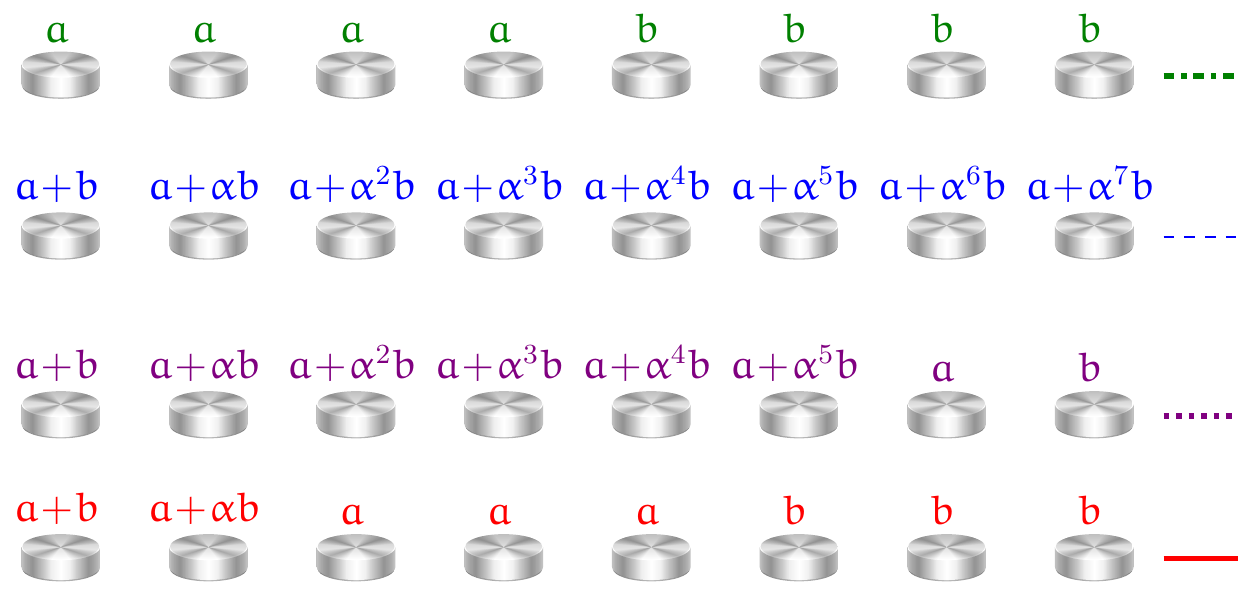}};
    \node at (8.35,0) {\includegraphics{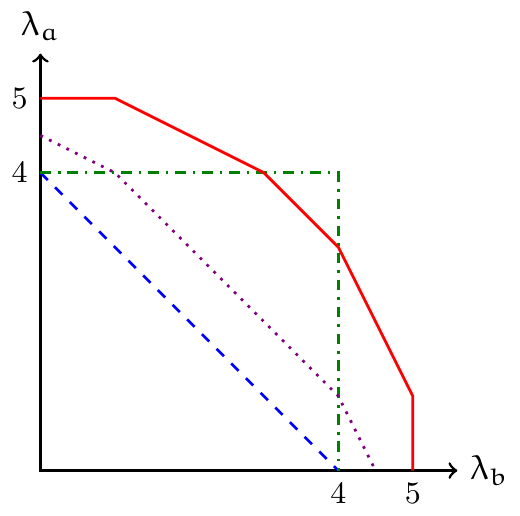}};
    \end{tikzpicture}
    \caption{Four coding schemes and their corresponding service rate regions. The largest region is achieved by combining coding and replication. ($\alpha$ is a primitive element of a sufficiently large finite field.)}
    \label{fig:MaxRegion}
\end{figure}
Their service rate regions are illustrated on the left side of the figure.
It is interesting to note that the service region depends on the encoding rather than on the code itself. In particular, the two codes in the middle (shown in blue and purple) are identical from a coding theory perspective, but their service rate regions are different. 
Amongst the four codes, we observe a combination of replication and coding (shown in red), where we create $3$ replicas each of $a$ and $b$ and $2$ coded combinations $a+b$ and $a+\alpha b$ can achieve the largest (by area) service rate region. Recent work \cite{ServiceCapacity:AktasAJ17} described the service rate region when there are $A$ replicas of object $a$, $B$ replicas of object $b$, and $C$ coded combinations of $a$ and $b$. Given the total number of servers $n = A+B+C$, determining the optimal values of $A$, $B$ and $C$ that maximize the area of the service rate region is an ongoing research direction. More generally, designing the generator matrix of a code to maximize the area/volume of the service rate region is an open problem. Also, since the service rate region is multi-dimensional, designing a fair metric other than the area/volume of the service rate region in order to compare the rate regions of two different classes of codes is an open problem. We propose some alternative metrics in \Cref{sec:open_problems}.
\subsubsection{Minimize the Number of Servers to Cover a Given Service Rate Region}
Consider a practical scenario where $k=2$ data objects, movies ``$a$'' and ``$b$'', are stored redundantly across multiple nodes in a coded storage system. At each time, each node can serve at most one request and each user can request to download at most one of the two movies $a$ and $b$. It is known that the number of users who are interested in downloading the movie $a$ and $b$ is less than or equal to $\alpha$ (i.e., $\lambda_a\le \alpha$) and $\beta$ (i.e., $\lambda_b\le \beta$), respectively. Also, it is known that the total number of users in the area is at most $\gamma$ (i.e., $\lambda_a+\lambda_b\le \gamma$). This means that the desired service rate region of this storage system is a bounded set $\mathcal{R}$ defined as follows:
\begin{equation}
\label{eq:regionconstrain}
    \mathcal{R}=\left\{\lambda_a,\lambda_b \ge 0, 
    \lambda_a\le \alpha,\lambda_b\le \beta,\lambda_a+\lambda_b\le \gamma\right\}.
\end{equation}

Two natural questions that arise in the design of this distributed storage system are the following: 1) What is the minimum number $n(\mathcal{R})$ of nodes required to serve all request vectors $\left(\lambda_a,\lambda_b\right)$ in the set $\mathcal{R}$?
2) How should the files $a$ and $b$ be stored redundantly in $n(\mathcal{R})$ storage nodes (i.e., what is the most storage-efficient redundancy scheme)? 
%
Using the example shown in \Cref{fig:CoverRegion}, we briefly illustrate how the storage-minimizing scheme varies with the shape of the service rate region that we wish to cover. 
\begin{figure}[htb]
\centering
\begin{tikzpicture}
\node at (-4,3) {\includegraphics[scale=1]{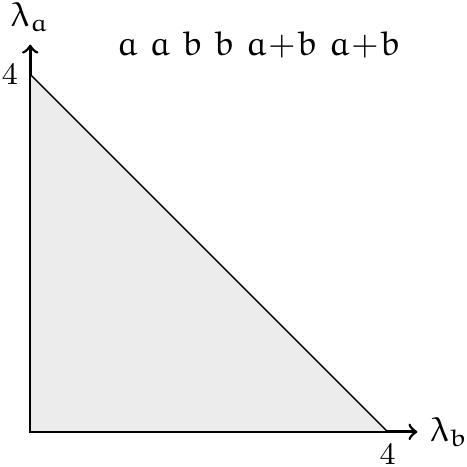}};
\node at (-4,0.25) {(a) $\alpha=4$, $\beta=4$ and $\gamma=4$};
\node at (4,3) {\includegraphics[scale=1]{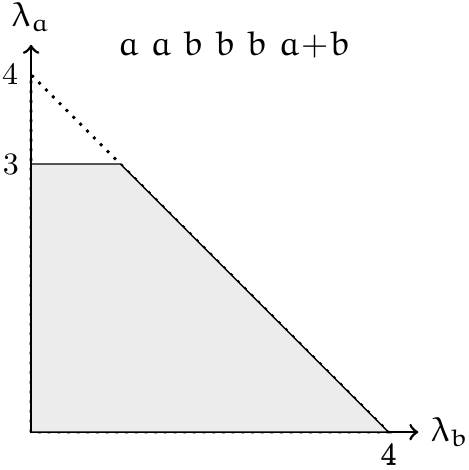}};
\node at (4,0.25) {(b) $\alpha=3$, $\beta=4$ and $\gamma=4$};
\node at (-4,-3) {\includegraphics[scale=1]{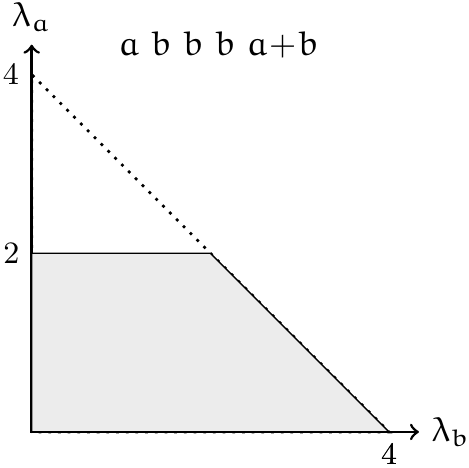}};
\node at (-4,-5.75) {(c) $\alpha=2$, $\beta=4$ and $\gamma=4$};
\node at (4,-3) {\includegraphics[scale=1]{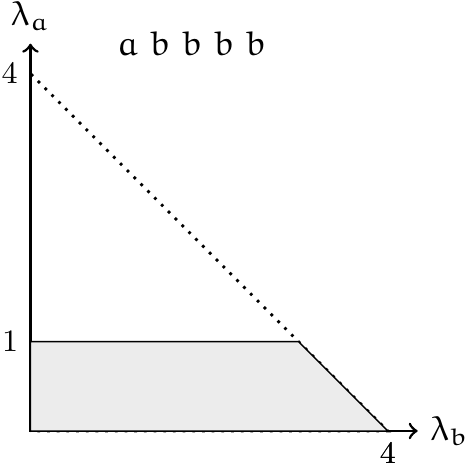}};
\node at (4,-5.75) {(d) $\alpha=1$, $\beta=4$ and $\gamma=4$};
\end{tikzpicture}
  \caption{Four service rate regions defined by the constraints $\lambda_a,\lambda_b \ge 0, 
    \lambda_a\le \alpha,\lambda_b\le \beta,\lambda_a+\lambda_b\le \gamma$, and their corresponding storage schemes that cover them with a minimum number of nodes. 
    } 
    \label{fig:CoverRegion}
\end{figure}
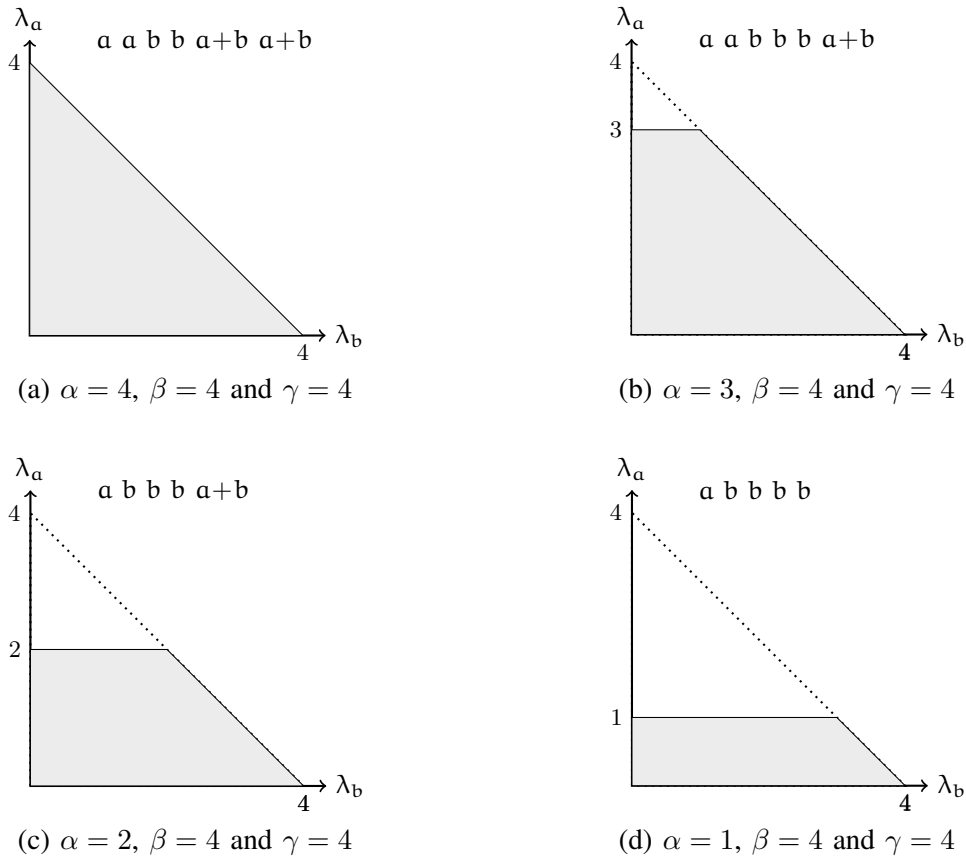
Let $\beta=4$ and $\gamma=4$, and $\alpha \in \{1,2,3,4\}$. The corresponding four storage-minimizing redundancy schemes (one for each $\alpha$) together with their service rate regions are shown in~\Cref{fig:CoverRegion}. In \Cref{fig:CoverRegion}(a), the rate region is dominated by points $(\lambda_a, \lambda_b)$ for which the demands for $a$ and $b$ are complementary to each other, that is, if $\lambda_a$ is high then $\lambda_b$ is low, and vice-versa. In this case, adding two coded nodes $a+b$ is the most storage-efficient way for achieving the service rate region. On the other hand, in \Cref{fig:CoverRegion}(d), where the demand for movie $b$ dominates the total request rate $\lambda_a + \lambda_b$, the best storage scheme does not have any coded nodes; it simply replicates object $b$ four times, and keeps just one uncoded copy of $a$.

In general, the problem of minimizing the number of nodes required for covering a desired service rate region, can be formulated as an integer linear programming (ILP). It is a challenging problem because in order to list all the constraints of the related ILP, one needs to explicitly know all possible recovery sets which becomes increasingly complex when the number of files $k$ increases. Recently, this problem of designing storage-efficient schemes to cover a desired service rate region has been studied for the first time in~\cite{kazemi2020efficient}, but there are still many open problems. For further details, please see~\cite{kazemi2020efficient}.

\section{Storage Schemes Considered in this Paper}
\label{sec:storage_schemes}
Each of the $n$ nodes in the system stores a linear combination of $k$ data objects which are mathematically represented as elements of the finite field $\mathbb{F}_q$.
We refer to a linear combination of data objects as coded object. If the linear combinations involves only one object, we call it systematic. We use the same terminology for the corresponding storage nodes. 
We refer to a linear code over a finite field $\mathbb{F}_{q}$ with block-length $n$ and dimension $k$ 
as an $[n,k]$ code. The generator matrix $G$ of a linear code is an $k \times n$ size matrix whose rows are a basis of the code, and their linear combinations form the codewords. 
We focus our attention to three classes of codes that are well-known in coding theory, see, e.g., \cite{MacWilliams-Sloane:78}.



\subsection{Maximum-distance-separable (MDS) Codes}
MDS codes achieve the well-known Singleton upper bound on the minimum distance, $d_{\min} \leq n-k+1$, hence the name. 
For an $[n,k]$ linear MDS code, any $k$ columns of the generator matrix are linearly independent, and thus  all the $k$ data objects can be recovered from any $k$ encoded objects. Therefore, for a systematic MDS code, the minimal recovery sets of a systematic column are the column itself and any $k$ of the remaining $n-1$ columns. An example of MDS codes commonly used in distributed storage systems is Reed Solomon codes. 

\subsection{Simplex Codes}
\label{sec:simplex}
A binary {\it Simplex code} (aka Hadamard code in CS literature) is a $[2^k-1, k]$ code with a generator matrix consisting of all distinct nonzero vectors of $\mathbb{F}_2^k$.\footnote{Although Simplex codes can be defined over any finite field, we restrict our attention to Simplex codes over the binary field in this paper.} Note that any generator matrix of a Simplex code has this form.
Simplex codes are useful in distributed storage systems, since each symbol of a Simplex code has $t = 2^{k-1}-1$ disjoint recovery sets of size two each~\cite{BoundsOnSizeOfLRCs:CadambeM15}. They are known to be optimal in several ways:
i) they meet the upper bound on the distance of codes having recovery sets of size at most two \cite{BoundsOnSizeOfLRCs:CadambeM15};
ii) they achieve the maximum storage efficiency among the binary linear codes with a given number of disjoint recovery sets of size two \cite{RateOptimalityOfSimplex:KadheC17};
iii) they meet the Griesmer bound and are therefore linear codes with the lowest possible length given the code distance \cite{Klein2004griesmer}.
Simplex codes play an important role in Computer Science as well, where they are known as Hadamard codes. 

\subsection{First Order Reed-Muller (RM) Codes}\label{gen-reedmuller}

A $k$-dimensional binary first-order Reed-Muller code $\text{RM}_2(1,k-1)$ with parameter $k \ge 2$, is a linear $[2^{k-1},k]$ code~\cite{muller1954application,reed1953class,assmus1994designs,arikan2009channel}. RM codes are important in both theory and practice. For a given $k$, the generator matrix of $\text{RM}_2(1,k-1)$ can be constructed as follows. 

Denote the set of all $(k-1)$-dimensional binary vectors by $\mathbb{F}_2^{k-1}=\{\mathbf{x}_1,\dots,\mathbf{x}_n\}$ where $n=2^{k-1}$ and for $i \in \{1,\dots,n\}$, $\mathbf{x}_i=(x_{i,{k-1}},\dots,x_{i,{1}})$ with $x_{i,j} \in \mathbb{F}_2$, $j \in \{1,\dots,k-1\}$. For any ${\mathcal{A} \subseteq \mathbb{F}_2^{k-1}}$, define the indicator vector $\mathbb{I}_\mathcal{A} \in \mathbb{F}_2^{k-1}$ as follows:
\[	
    (\mathbb{I}_\mathcal{A})_i=	
    \begin{cases}
	1 & \text{if }\mathbf{x}_i\in \mathcal{A}\text{;} \\
	0 & \text{otherwise.}
	\end{cases}
\]

For the $k$ rows of the generator matrix of $\text{RM}_2(1,k-1)$, define $k$ row vectors of length $2^{k-1}$ as follows, $\mathbf{r}_0=(1,\dots,1)$ and ${\mathbf{r}_j=\mathbb{I}_{\mathcal{H}_j}}$, where ${j \in \{1,\dots,k-1\}}$ and ${\mathcal{H}_j=\{\mathbf{x}_i \in \mathbb{F}_2^{k-1} \mid x_{i,j}=0\}}$. The set $\{\mathbf{r}_{k-1},\dots,\mathbf{r}_1,\mathbf{r}_0\}$ defines the rows of a non-systematic generator matrix of the $\text{RM}_2(1,{k-1})$. For a systematic generator matrix of $\text{RM}_2(1,{k-1})$, the set of rows $\{\mathbf{r}_{k-1},\dots,\mathbf{r}_1,\sum_{i=0}^{k-1}\mathbf{r}_i\}$ can be considered.

\section{Service Rate Region Using Waterfilling}
\label{sec:MDS}
In this section we find the service rate region of a system of $n$ servers that store $k$ data objects $u_1, \ldots, u_k$ using maximum-distance-separable codes or locally recoverable codes. Suppose that we use an $[n,k]$ systematic MDS code to generate the data stored on each of the $n$ servers. Each object $u_i$ can be downloaded from the server storing it, which we refer to as the systematic server, or by accessing any $k$ of the remaining $n-1$ servers. Let the arrival rate of requests for object $u_i$ be $\lambda_i$. We want to determine the set of arrival rate vectors $(\lambda_1, \dots \lambda_k)$ that can be supported by the system. In other words, without loss of generality, we want to maximize $\lambda_k$ for any given a feasible set of rates $(\lambda_1, \dots , \lambda_{k-1})$.

We seek to find a strategy to split the download requests across the $n$ servers. We now propose the following {\it water-filling} algorithm to split the request rate among the $n$ servers of an $(n, k)$ coded system. The high-level idea behind this algorithm is that requests are first routed to the respective uncoded or systematic server. Once the systematic servers are saturated, the requests are sent to the $k$ least-loaded servers that have not been yet saturated by $\mu$ rate of requests. Below, we present the water-filling algorithm below for MDS and locally recoverable codes and show that its resulting request allocation achieves the optimal service rate region for MDS codes. However, the core idea of waterfilling-based request splitting is broadly applicable beyond these two classes of codes.

\begin{definition}[Waterfilling Algorithm for MDS Coded Systems]
\label{defn:mds_waterfilling}
Assume that the request arrival rates $\lambda_1$, $\lambda_2$, \dots $\lambda_k$ for the $k$ objects are $\lambda_1 \geq  \lambda_2 \geq \dots \geq \lambda_k$ without loss of generality. Let $\gamma_i$ denote the assigned load, or the request rate assigned to server $i$. The water-filling algorithm assigns them to the $n$ servers as follows.
\begin{enumerate}
\item \textbf{Assign requests to systematic (uncoded) nodes.} We first assign the arrival rate $\lambda_i$ for object $i$ to the respective systematic (uncoded) server until that server is saturated. Thus, the $i^{th}$ systematic server gets assigned the load $\gamma_i = \min(\lambda_i, \mu)$ for $i =1, \dots k$. The remaining $\sum_{i=1}^k (\lambda_i - \min(\lambda_i, \mu)) = \sum_{i=1}^k  (\lambda_i - \mu)^+$ arrival rate needs to be served using the $(n-k)$ coded nodes and the unsaturated systematic nodes. 
\item \textbf{Assign each request to the $k$ least-loaded nodes.} The remaining $\lambda_{coded} = \sum_{i=1}^k(\lambda_i - \mu)^+$ load is split across the nodes in the following manner. 
%
While $\lambda_{coded}  > 0$ and $\min_i \gamma_i < \mu$ do the following:
\begin{itemize}
\item Find the set $\mathcal{S}$ of the $k$ least-loaded servers (with minimum $\gamma_i$) in the system. If there are more than $k$ servers with the same minimum $\gamma_i$, choose $k$ servers uniformly at random.
\item Assign a small rate $\epsilon>0$ of requests to these $k$ least-loaded servers 
\item Decrement $\lambda_{coded}$ by $\epsilon$, that is, $\lambda_{coded} \leftarrow \lambda_{coded} - \epsilon$
\item Increment the corresponding $k$ server loads by 
$\epsilon$, that is,  $\gamma_i \leftarrow \gamma_i + \epsilon$ for all $i \in \mathcal{S}$.
\end{itemize}
\end{enumerate}
\end{definition}

The algorithm is illustrated in \Cref{fig:mds_waterfill_illustration} for a $(6,3)$ MDS code. After sending the requests to their respective systematic nodes, nodes $1$ and $2$ are saturated but since $\lambda_3 < \mu$, the third systematic node has some remaining service capacity, which can be used to serve the overflow of requests for objects $1$ and $2$. The total overflowing request rate for data objects $1$ and $2$ is $\lambda_{\text{coded}} = (\lambda_1 - \mu)^+ + (\lambda_2 - \mu)^+$ left to be served. These requests can be served by accessing any $k=3$ of the unsaturated nodes in the system, decoding all $k=3$ data objects, and obtaining the object of interest. Since all $k = 3$ objects end up being decoded, we do not need to consider the overflowing requests for objects $1$ and $2$ separately, but consider them together as $\lambda_{\text{coded}} = (\lambda_1 - \mu)^+ + (\lambda_2 - \mu)^+$. 

We first send $\lambda_3$ out of the $\lambda_{\text{coded}}$ rate (shown by the light green and blue shaded regions in \Cref{fig:mds_waterfill_illustration}) to the MDS coded nodes $4$, $5$, $6$, since these are the $k$ least loaded nodes in the system. After this allocation, servers $3$, $4$, and $5$, $6$ all have $\mu - \lambda_3$ capacity left and are the least-loaded nodes in the system. Any $3$ of these $4$ servers can be used to serve the remaining $\lambda_{\text{coded}} - \lambda_3$ rate of requests (shown in grey color). Each coded request will be served by accessing $3$ coded data objects and then decoding the object of interest. There are  $\binom{4}{3} = 4$ sets of $3$ servers each, and each server participates in $3$ such sets. Thus, the additional load allocated to servers $3$, $4$, $5$, $6$ (shown in grey) is $3(\lambda_{\text{coded}} - \lambda_3)/4$ rate each. Here we need the additional load $3(\lambda_{\text{coded}} - \lambda_3)/4$ to be less than $\mu - \lambda_3$, the remaining capacity of each of the nodes $3$, $4$, $5$ and $6$.

\begin{figure}[t]
    \centering
   \includegraphics[width= 1.00\textwidth]{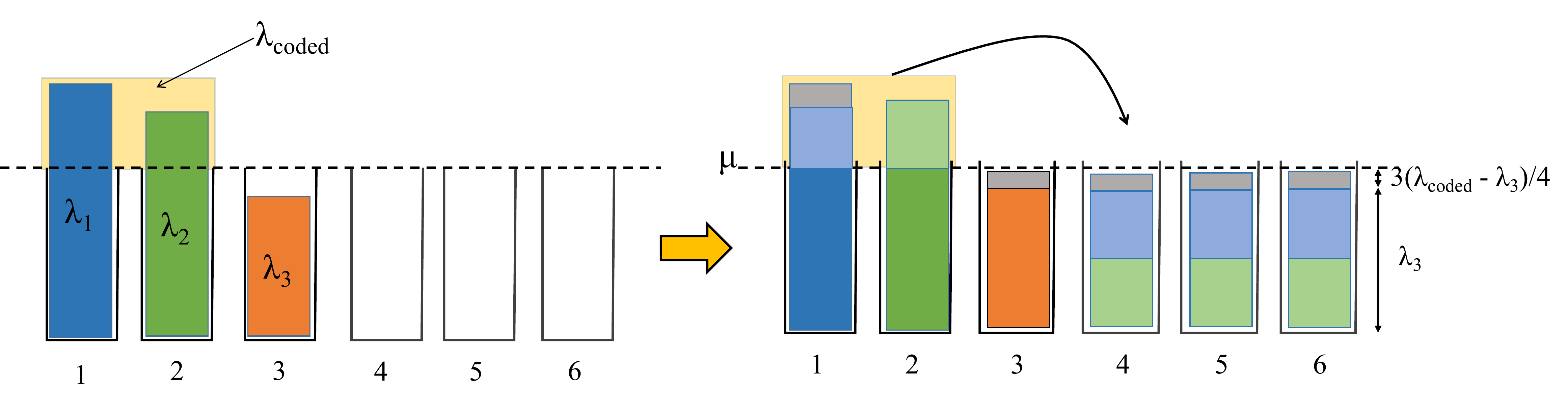}
\caption{Water-filling strategy to serve the requests using coded nodes for a $(6,3)$-MDS code.\label{fig:mds_waterfill_illustration}}
\vspace{-0.3cm}
\end{figure}

\subsection{Service Rate Region for MDS codes}
Below we first find a converse or upper bound on the achievable service rate region of MDS codes. 

\begin{theorem}
\label{thm:converse_bnd}
The set of all achievable request vectors $(\lambda_1, \lambda_2, \dots, \lambda_k)$ of an $(n,k)$ systematic-MDS coded system lies inside the region described by 
\begin{align}
      \sum_{i=1}^{k} \left( \min(\lambda_i, \mu) +  k(\lambda_i - \mu)^{+} \right) \leq n \mu, \label{eqn:mds_outer_bnd}
\end{align}
\end{theorem}

\begin{proof}
To prove this outer bound on the achievable rate region, observe that each server in the system can support $\mu$ requests/time, and thus the total capacity is $n \mu$. Downloading each data object from its own systematic (uncoded) node uses only $1$ unit of capacity. However, downloading an object from $k$ coded servers requires $k$ units of capacity per unit request rate. Thus, if $\lambda_i$ is the rate of request arrivals for object $i$, the minimum system capacity utilized by these requests is  $\min(\lambda_i, \mu) +  k(\lambda_i - \mu)^{+}$, where $\min(\lambda_i, \mu)$ requests are served by the systematic node storing object $i$. Since the total system capacity is $n \mu$, the sum of the capacity utilized by all requests must be less than $n \mu$. Thus we have \eqref{eqn:mds_outer_bnd}.
\end{proof}

Next, we show that this water-filling algorithm is optimal, that is, it can serve any achievable set of request rates $(\lambda_1, \dots \lambda_k)$. To prove the optimality we separately consider two cases below: 1) $n-k \geq k$ (the code rate $\leq 1/2$), and 2) $n-k < k$ (the code rate  $> 1/2$). 

\begin{theorem}
\label{thm:waterfilling_opt}
The water-filling algorithm proposed in \Cref{defn:mds_waterfilling} is optimal, that is, it achieves the outer bound given by \eqref{eqn:mds_outer_bnd}, for any MDS code when $n-k \geq k$. 
\end{theorem}
\begin{proof}

For $n-k \geq k$ we now evaluate the set of arrival rates that can be achieved by the waterfilling algorithm and show that it matches the outer bound in \eqref{eqn:mds_outer_bnd}. Without loss of generality, sort the arrival rates in descending order such that $\lambda_1 \geq \lambda_2 \geq \dots \geq \lambda_k$. After sending requests to systematic servers until they are saturated, the total residual arrival rate is $\lambda_{coded} = \sum_{i=1}^k(\lambda_i - \mu)^+$, as illustrated in \Cref{fig:mds_waterfill_illustration} for the $(6,3)$ MDS coded system. Assume that $\lambda_1 \geq \mu$. If this is not true, then $\lambda_{coded} = 0$ and all requests can be served by systematic servers.

The waterfilling algorithm first uniformly splits requests over $n-k$ coded nodes, $k+1 ,\dots n$. We do this until the load at all these nodes becomes equal to $\lambda_k$, the least-loaded systematic node. Since each request needs be sent to $k$ out of the $n-k$ servers, up to $\min(\lambda_k, \mu) (n-k)/k$ requests can be served in this manner. After this assignment, there are $n-k+1$ nodes from node $i = k ,\dots n$ with the same load $\gamma_i = \min(\lambda_k, \mu)$. The waterfilling algorithm now assigns each request to the least-loaded $k$ out of these $n-k+1$ servers until their load reaches $\min(\lambda_{k-1}, \mu)$, the load of the $(k+1)^{th}$ least-loaded server. The request rate that be assigned this way is $(\min(\lambda_{k-1}, \mu)- \min(\lambda_k, \mu)) (n-k+1)/k$. Recursively repeating this process for every $r = k, \dots, 2$, we uniformly split $\min( (\gamma_{r-1}- \gamma_{r})(n-r+1)/k, \lambda_{coded})$ requests over $n-r+1$ servers, $r, r+1, \dots n$. Thus, using the proposed water-filling algorithm, the maximum rate $\lambda_{coded}^{max}$ of requests that can be supported using coded servers is
\begin{align}
\lambda_{coded}^{max} &= \min(\lambda_k, \mu) \frac{n-k}{k} + (\min(\lambda_{k-1},\mu) - \min(\lambda_{k}, \mu))\frac{n-k+1}{k}+ \dots + \nonumber\\
& \quad \quad \quad \quad (\min(\lambda_{1},\mu) - \min(\lambda_{2}, \mu))\frac{n-1}{k} \\
& = \min(\lambda_1, \mu) \frac{n}{k} - \sum_{i=1}^{k} \min(\lambda_i, \mu) \frac{1}{k} \\
&= \mu \frac{n}{k} - \sum_{i=1}^{k} \min(\lambda_i, \mu) \frac{1}{k}
\end{align}
In~\Cref{fig:mds_waterfill_illustration}, the height of each lightly-shaded portion, starting from the bottom upwards, corresponds to each term in the above summation. 

After saturating the systematic nodes, the residual rate $\lambda_{coded} = \sum_{i=1}^k(\lambda_i - \mu)^+$ supported by the coded servers can be at most $\lambda_{coded}^{max}$. That is,
\begin{align}
 \lambda_{coded} &\leq \lambda_{coded}^{max} \\
 \sum_{i=1}^{k}(\lambda_i - \mu)^+ &\leq  \mu \frac{n}{k} - \sum_{i=1}^{k} \min(\lambda_i, \mu) \frac{1}{k}
\end{align}
Rearranging, this is equivalent to \eqref{eqn:mds_outer_bnd}. Thus, for $n-k \geq k$, waterfilling can achieve the region given by the outer bound in \Cref{eqn:mds_outer_bnd}. Hence, the proposed waterfilling algorithm is optimal for $n-k \geq k$.
\end{proof}

Next let us consider the second case $n-k < k$. For this case, we cannot always achieve the the same rate region as given by the outer bound in \eqref{eqn:mds_outer_bnd}. However, we can show that the waterfilling algorithm is optimal, and no other rate splitting scheme can yield a strictly larger rate region. 

\begin{lemma}
\label{lem:waterfilling-high-rate-optimality-1}
It is optimal to first send requests to their systematic node. Only when the systematic node is saturated, requests should be served using coded servers.
\end{lemma}
\begin{proof}
Suppose $\lambda_i < \mu$ for some $i$, that is all requests for object $i$ can be served by the systematic node. Instead, suppose we serve $\lambda_i - \epsilon$ rate using the systematic node $i$, and send the remaining $\epsilon$ portion to $k$ other servers, and decode file $f_i$ from the coded versions. As a result we are reducing the load on the systematic node by $\epsilon$, and instead adding $\epsilon$ load to $K$ other servers. If $n-k < k$, at least one of these $k$ servers is also a systematic node, which stores file $f_j$. Thus, the maximum rate of requests for file $f_j$ that can be served by its systematic node reduces by $\epsilon$. For $n-k > k$, we showed in \Cref{thm:waterfilling_opt} that the water-filling algorithm, which first sends requests to the systematic node is optimal. Thus, there is no loss of optimality in sending requests to the systematic node until it is saturated.
\end{proof}

\begin{lemma}
\label{lem:waterfilling-high-rate-optimality-2}
After the systematic node is saturated, it is optimal to always send each request to the $k$ least-loaded servers that can serve it.
\end{lemma}
\begin{proof}
Each $\epsilon > 0$ portion of the request rate $\lambda_{coded}$, needs to be allocated to $k$ servers. Using any algorithm for picking the $k$ servers, we could reach one of the two possible states:
\begin{enumerate}
\item $r \geq k$ unsaturated servers with the same load $\gamma < \mu$. Then we can split a maximum of $(\mu - \gamma)r/k$ request rate uniformly over these servers. As a result all servers will be saturated, and the outer bound on the service rate region can be achieved.
\item There are exactly $k$ unsaturated servers in the system with loads $\gamma_1 \geq \gamma_2 \geq \gamma_3 \geq \dots \geq \gamma_k$, where at least one of these inequalities is strict. The waterfilling can serve an additional $\mu -\gamma_1$ rate of requests. This would leave a non-zero amount of capacity unused.
\end{enumerate}
Since water-filling algorithm always sends requests to the $k$ least-loaded nodes in the system, observe that it achieves the first state whenever it is feasible. And if the system ends up in the second state, water-filling minimizes the total unused system capacity $n \mu - \sum_{i=1}^{n} \gamma_n$ by always allocating requests to the least-loaded servers. 
\end{proof}



\subsection{Service Rate Region for Locally Recoverable Codes}
\label{sec:LRCs}
Locality of a code captures the number of symbols participating in recovering a lost symbol. In particular, an $[n,k]$ code is said to have locality $r$ if every symbol is recoverable from a set of at most $r$ symbols. For linear codes with locality, a {\it local} parity check code of length at most $r+1$ is associated with every symbol. 
The notion of locality can be generalized  to accommodate {\it local codes} of larger distance as follows (see~\cite{Prakash:12}). For an $[n,k]$ code $\mathcal{C}$ and a subset $S\subset[n]$, we use $\mathcal{C}_S$ to denote $\mathcal{C}$ restricted to symbols in $S$.

\begin{definition}
\label{def:locality}
[Locality] An $[n, k]$ code $\mathcal{C}$ is said to have $(\ell,r)$ information locality $(\ell > r)$, if for every data object $i$, there exists a set of indices $\Gamma_{i}$ such that
(i) $i\in\Gamma_i$,
(ii) $|\Gamma_i| \leq \ell$, and
(iii) $d_{\min}(\mathcal{C}_{\Gamma_i}) \geq \ell - r + 1$.
The code $\mathcal{C}_{\Gamma_i}$ is said to be the local code associated with the $i$-th data object.
\end{definition}

Properties 2 and 3 imply that for any codeword in $\mathcal{C}$, the values in $\Gamma_{i}$ are uniquely determined by any $r$ of those values. Therefore, the $(\ell, r)$ locality allows one to {\it locally} repair any $\ell - r$ erasures in $\mathcal{C}_{\Gamma_i}$, $\forall i\in[n]$, by accessing $r$ other objects. When $\ell = r+1$, the above definition reduces to the classical definition of locality proposed by Gopalan et al.~\cite{Gopalan:12}, wherein any one erasure can be repaired by accessing at most $r$ objects.

Throughout the rest of this section, we focus on LRCs that have the same structure as the Pyramid code from~\cite{Prakash:14}. In particular, the $k$ data objects are partitioned into $k/r$ groups, and each group has $\ell - r$ local parities satisfying the properties of Definition~\ref{def:locality}. Each such group is called a local group. Further, the code has $p$ global parities. 
\begin{example}
\label{ex:LRC}
Consider an $(12,4)$ LRC with $(4,2)$ locality and $p = 4$ global parities as follows:
$$
\left[
\begin{array}{cccccccccccc}
{a} & {b} & {c} & {d} & {a+b} &
{c+d} & {a+2b} &  {3c+4d} & 
p_1 & p_2 & p_3 & p_4
\end{array}
\right],
$$
where $p_i$, $1\leq i\leq 4$ denote global parity symbols. Observe that $[a\: b\: a+b\: a+2b]$ and $[c\: d\: c+d\: c+2d]$ are local groups.
\end{example}


Next, we generalize the waterfilling algorithm to LRCs. One key difference than MDS codes is that in LRCs it is not possible to recover all the $k$ data objects from any $k$ coded objects. In particular, each local group has $r$ linearly independent symbols, and sending a request to more than $r$ servers in a local group is redundant. Further, some set of $k$ servers cannot recover all the data symbols.\footnote{LRCs which have the information-theoretically optimal recovery guarantees are referred to as maximally recoverable codes. See~\cite{Gopalan:14} and references therein.} For instance, for Example~\ref{ex:LRC}, any one parity server from each of the two local groups together with any two global parity servers cannot be used to recover the four data objects. On the other hand, it is not difficult to see that in parity-splitting LRCs like Pyramid codes, one can recover the $k$ data objects from any $r$ parity symbols for $\ell$ local groups, where $1\leq\ell\leq\lceil k/r\rceil$, and any $k - r\ell$ global parity symbols. We restrict to such sets of servers in the final step of waterfilling. 

\begin{figure}[t]
    \centering
   \includegraphics[width= 0.98\textwidth]{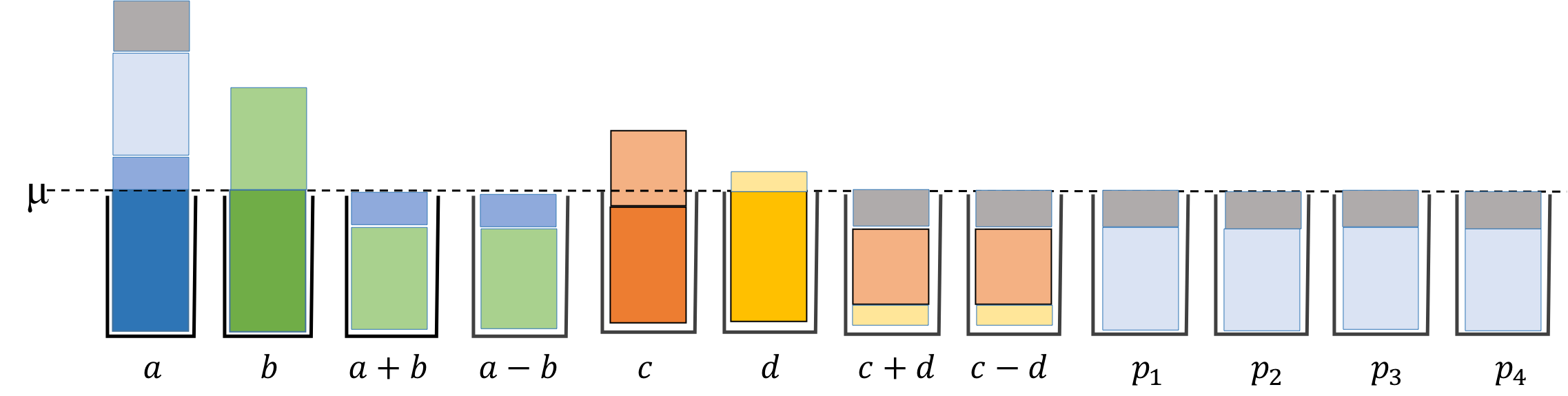}
\caption{Water-filling strategy to serve the requests using coded nodes for a $(12,4)$-LRC code.\label{fig:LRC_waterfill_illustration}}
\vspace{-0.3cm}
\end{figure}


\begin{definition}[Waterfilling Algorithm for LRC Coded Systems]
\label{defn:mds_lrc_waterfilling}
Assume that the request arrival rates $\lambda_1$, $\lambda_2$, \dots $\lambda_k$ for the $k$ objects are $\lambda_1 \geq  \lambda_2 \geq \dots \geq \lambda_k$ without loss of generality. 
Let $\gamma_i$ denote the assigned load, or the request rate assigned to server $i$. 
\begin{enumerate}
\item \textbf{Assign requests to systematic (uncoded) nodes and coded nodes in local groups using waterfilling as in Definition~\ref{defn:mds_waterfilling}.} 
\item \textbf{Assign each request to $k$ least loaded coded nodes.} 
Denote the remaining load as $\lambda_{coded}$. 

While $\lambda_{coded}  > 0$ and $\min_i \gamma_i < \mu$ do the following:
\begin{itemize}
\item Find a set $\mathcal{S}$ of $k$ least loaded servers such that if the set contains a parity server from any local group, then it should contain $r$ servers from the same group. If there are multiple such sets of $k$ least loaded servers, choose a set uniformly at random. 

\item Assign a small rate $\epsilon>0$ of requests to every server in $\mathcal{S}$, 
increment the corresponding server loads by 
$\epsilon$, and 
decrement $\lambda_{coded}$ by $\epsilon$.

\end{itemize}
\end{enumerate}
\end{definition}

Unlike MDS codes for which the waterfilling algorithm is optimal, it is open whether the waterfilling algorithm is optimal for LRC codes. The techniques used for proving the optimality of the waterfilling algorithm for MDS codes are not sufficient for analyzing LRC codes. This is because, after the systematic nodes are saturated, a request can be satisfied using three ways: (i) only the local parity nodes, (ii) a mix of local and global parity nodes, and (ii) only the global parity nodes (when $n-k \geq k$). Thus, the recovery sets for LRCs are more complex (as opposed to $k$-node subsets for MDS codes), which calls for novel techniques to analyze the waterfilling algorithm for LRCs.


\subsection{Summary}\label{sec:waterfilling-summary} 
In this section, we introduced the waterfilling strategy to determine how to split requests across different nodes in a coded distributed system. We analyzed it for MDS and locally recoverable codes, and showed that it is optimal for MDS codes. Proving its optimality for LRCs remains open for the reasons that we described above. Simplex codes can also be considered as LRCs with $(3,2)$ locality~\cite{BoundsOnSizeOfLRCs:CadambeM15,RateOptimalityOfSimplex:KadheC17}. Therefore, one can use the waterfilling algorithm to allocate requests in a Simplex coded system as well. However, it is not clear how to use waterfilling arguments to characterize the service rate region for Simplex codes in a closed form. As we will see next, the other two approaches, combinatorial optimization and geometric, are well suited to characterize the service rate region for the Simplex codes. In general, not surprisingly, each approach is well suited for specific types of codes.

More broadly, the waterfilling strategy encompasses two key ideas. Firstly, it takes into account the fact that each request is associated with a ranked preference list of subsets of servers that it wants to be assigned to. For example, in MDS coded systems sending a request to a systematic node preferred over sending it to $k$ coded nodes. Secondly, at each preference level, waterfilling assigns the request to the least loaded server(s) in order to maximize the achievable service rate region. Although we propose it in the context of coded storage systems, these central ideas of the waterfilling strategy can be utilized more broadly for resource allocation in distributed systems. Exploring other applications of this strategy is an interesting and open future direction.

\section{Service Rate Region Using Combinatorial Optimization on Graphs}
\label{sec:simplex_comb_opt}

In this section, we introduce a graph representation of a coding scheme, and look at the service rate region problem through the lens of combinatorial optimization on graphs. We begin with briefly reviewing the notions of matching and vertex cover in graphs. For details, we refer the reader to standard texts on graph theory, e.g.,~\cite{West:01}.

\subsection{Matching and Vertex Cover on Graphs}
\label{sec:graph-preliminaries}


A {\it matching} in graph $\Gamma$ is a set of pairwise non-adjacent edges. A maximum matching in $\Gamma$ is a matching that contains the largest number of edges. The size of a maximum matching is known as the {\it matching number}, and it is denoted as $\nu(\Gamma)$. Note that a matching can be considered as assigning to each edge a weight from the set $\{0,1\}$ such that the sum of the weights on the edges incident on any vertex is at most one.

A {\it fractional matching} allows one to assign any fraction in the interval $[0,1]$ as a weight to each edge such that the sum of the weights on the edges incident on any vertex is at most one. A maximum fractional matching of $\Gamma$ has the maximum sum of weights among all the fractional matchings of $\Gamma$. The sum of weights of a maximum fractional matching is known as the {\it fractional matching number}, and it is denoted as $\nu_f(\Gamma)$.

A {\it vertex cover} of a graph $\Gamma$ is a set of its vertices such that each edge in $\Gamma$ is incident to at least one vertex in the set. A minimum vertex cover is a vertex cover of smallest possible size. The size of a minimum vertex cover is known as the {\it vertex cover number}, and it is denoted as $\tau(\Gamma)$. For any graph $\Gamma$, it holds that ${\nu(\Gamma) \leq \nu_f(\Gamma) \leq \tau(\Gamma)}$, and, in particular for a bipartite $\Gamma$, we have ${\nu(\Gamma) = \nu_f(\Gamma) = \tau(\Gamma)}$.
 
\subsection{Graph Representation of Storage Schemes}\label{sec:graph_representations} Here, we introduce a graph representation of storage schemes described in \Cref{sec:storage_schemes}. For simplicity, we consider linear codes, however, it is straightforward to generalize the notion for non-linear codes. For the clarity of exposition, we focus on recovery sets of size one and two. In other words, a recovery set for each object is either a systematic symbol or a group of two symbols, as is the case when $k=2$  and for Simplex codes of any dimension. As we discuss in Remark~\ref{rem:hypergraphs} later, the notions described next can be easily extended to the general case of arbitrary sized recovery sets by considering hypergraphs. We consider hypergraphs associated with MDS codes in Sec.~\ref{Sec:recoverygraphMDS}.

Consider an $[n,k]$ code with a $k\times n$ generator matrix $G$. We define a graph $\Gamma_G$ associated with the $G$ as follows.
$\Gamma_G$ has $n$ vertices corresponding to $n$ columns of $G$. For every recovery set (of size two) of data symbol $x$, the corresponding vertices in $\Gamma_G$ are connected by an edge with label $x$. We refer to such an edge as $x$-recovery edge. If $G$ is systematic, an additional vertex is added for each systematic column, and it is connected by an edge to the vertex corresponding to the systematic column and labeled accordingly. This avoids self-loops corresponding to recovery sets of size one formed by the systematic columns. We refer to $\Gamma_G$ as a \textit{recovery graph} for the coding scheme $G$.

\Cref{fig:CodeMatchGraph} shows generator matrices with their recovery graphs for a systematic $(4,2)$ MDS code and the $[7,3]$ Simplex code.
\begin{figure}[hbt]
\begin{center}
	\includegraphics[scale=0.99]{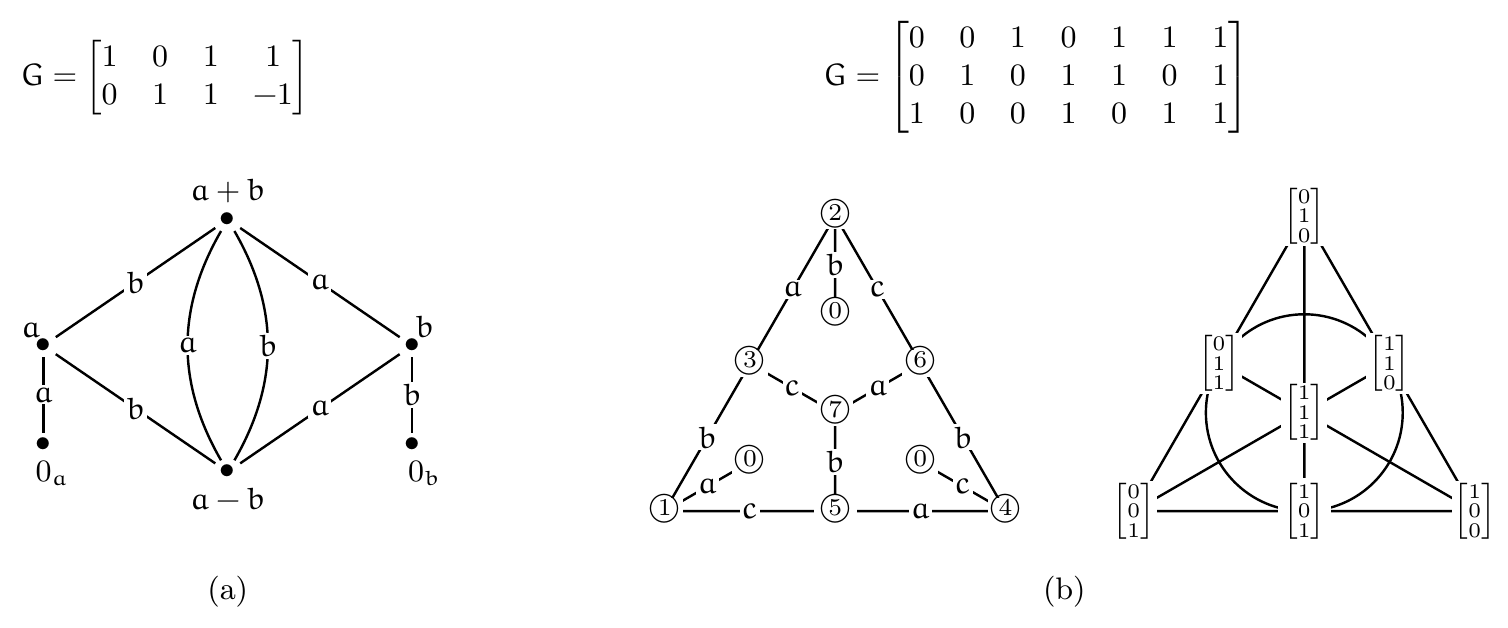}
\end{center}\vspace{-0.25cm}
    \caption{(a) Generator matrix and its recovery graph for a systematic $[4,2]$ MDS code. (b)~Generator matrix and its recovery graph for the $[7,3]$ Simplex code, and the Fano plane.}
    \label{fig:CodeMatchGraph}
\end{figure}
In Sec.~\ref{sec:simplex_comb_opt}, we show how the service rate problem associated with matrix $G$ is related to matching and vertex cover problems of its recovery graph $\Gamma_G$.

\subsection{Service Allocation as a Fractional Matching in the Recovery Graph}
\label{sec:fractional-matching}

Associating a recovery graph with a coding scheme allows us to relate the service allocation problem to the problem of finding a fractional matching in the recovery graph. 
(For the original observation and more details, see \cite{ServiceCombinatorial:KazemiKSS20}.)
Let us consider a coding scheme $G$ and its recovery graph $\Gamma_G$. 
We demonstrate that a valid allocation for the coding scheme $G$ for a given demand vector is equivalent to a fractional matching with specific constraints on $\Gamma_G$. In the rest of this section, we assume without loss of generality that $\mu = 1$. 

\begin{proposition}
\label{prop:fractional-matching}
Consider a system using an $[n,k]$ code with a generator matrix $G$ such that every recovery set is of size at most two and $\mu = 1$. The system can serve a demand vector $(\lambda_1, \lambda_2, \cdots, \lambda_k)$ if and only if there exists a fractional matching in the recovery graph $\Gamma_G$ such that the weights on the edges with label $i$ sum to $\lambda_i$.
\end{proposition}

\begin{proof}
Suppose there exists a fractional matching in $\Gamma_G$ such that the weights on the edges with label $i$ sum to $\lambda_i$. Let $w_{i,j}$ denote the weight on the edge corresponding to recovery set $R_{i,j}$. Then, for every object $i$, assign $w_{i,j}$ fraction of its load $\lambda_i$ to recovery set $R_{i,j}$ for $j\in[t_i]$. Since the sum of the weights of edges does not exceed one at any vertex, no server is assigned requests in excess of its service rate. Further, since the weights on the edges with label $i$ sum to $\lambda_i$, the demands for all objects are served. 

On the other hand, suppose that there is a valid allocation for a demand vector $(\lambda_1, \lambda_2, \cdots, \lambda_k)$. Let $\lambda_{i,j}$ denote the load assigned to recovery set $R_{i,j}$. Then, for every $i\in[k]$, assign the weight $\lambda_{i,j}$ for the edge labeled $i$ that is corresponding to recovery set $R_{i,j}$. Since the allocation $\{\lambda_{i,j}:1\leq i\leq k, 1\leq j\leq t_i\}$ satisfies~\eqref{eq:allocation-constraints_1}-\eqref{eq:allocation-constraints_3}, it is immediate to see that the weights assigned form a fractional matching such that the sum of the weights on the edges with label $i$ sum to $\lambda_i$. 
\end{proof}

\begin{remark}
\label{rem:hypergraphs}
While defining recovery graphs and in Proposition~\ref{prop:fractional-matching}, we restricted our attention to linear coding schemes having recovery sets of size at most two. The general case of a code having recovery sets of arbitrary size can be straightforwardly considered by associating a hypergraph with the code's generator matrix. 
Note that a hypergraph is a generalization of a graph in which any subset of vertices may be joined by an edge, called a hyperedge (see, e.g.,~\cite[Chapter 7]{Voloshin:09:hypergraph}). 
Specifically, we form a hypergraph $\Gamma_G$ associated with $G$ such that its vertices correspond to columns of $G$ and hyperedges correspond to recovery sets. It is straightforward to generalize the hypergraph representation for non-linear codes. See Sec.~\ref{Sec:recoverygraphMDS} for hypergraphs associated with MDS codes.
\end{remark}

The relation to fractional matching enables us to obtain bounds on the service rate region, and, in some cases, completely characterize the service rate region. First, we present a bound on the sum of the request rates that can be served by the system
using vertex covers in $\Gamma_G$. Recall that a vertex cover of a graph $\Gamma$ is a set of vertices of $\Gamma$ such that each edge in $\Gamma$ is incident to at least one vertex in the set. 
From Proposition~\ref{prop:fractional-matching} and the well-known combinatorial optimization result that the fractional matching number is upper bounded by the vertex cover number of any graph, we get the following upper bound on the sum of request rates that can be served by a system.

\begin{proposition}
\label{prop:vertex-cover}
Consider a system using an $[n,k]$ code with a generator matrix $G$, and let $\Gamma_G$ be the recovery graph of $G$. The sum of rates in any demand vector $(\lambda_1, \cdots, \lambda_k)$ that can be served by the system cannot exceed the number of vertices in a cover of $\Gamma_G$.
\end{proposition}

\subsection{Using Graph Representations to Characterize Service Rate Regions}

\subsubsection{Simplex Codes}\label{sec:sim} Here, we characterize the service rate region of binary Simplex codes. For clarity of exposition, we focus our attention to non-overlapping recovery sets of size two. Later, we show that considering all the recovery sets does not increase the service rate region. Note that, since any generator matrix of a Simplex code consists of all non-zero length-$k$ binary vectors, any generator matrix is a column permutation of the other. Thus, recovery graphs associated with all generator matrices of a Simplex code are isomorphic, and consequently, we refer to a recovery graph associated with arbitrary generator matrix of a Simplex code as the recovery graph of the Simplex code. The first step is to show that the recovery graph of a Simplex code is bipartite. (See~\Cref{fig:CodeMatchGraph} for an example of the $[7,3]$ Simplex code.) We note that has been shown in~\cite{ServiceCapacity:AktasAJ17,Wang:17:switch-codes,ServiceCombinatorial:KazemiKSS20}. We present a brief proof for completeness.
\begin{lemma}[Structure of the Recovery Graph for Simplex Codes]
\label{lem:simplex-bipartite}
For a $[2^{k}-1,k]$ Simplex code with recovery graph $\Gamma_k$, the following holds:
\begin{enumerate}
    \item $\Gamma_k$ is  bipartite.
    \item Each vertex of $\Gamma_k$ has degree $k$ where each edge corresponds to a recovery set of a different object.
    \item The $2^{k-1}$ vertices of $\Gamma_k$ that correspond to the odd weight columns of $G_k$ form a minimal vertex cover of $\Gamma_k$.
\end{enumerate}
\end{lemma}
\begin{proof}
Recall that a generator matrix $G_k$ of the $k$ dimensional Simplex code consists of all non-zero length-$k$ binary vectors. Let us label each of the $2^{k}-1$ vertices of its recovery graph $\Gamma_k$ with a length-$k$ non-zero binary vector. In addition, let us label each of the additional vertices for systematic columns by the length-$k$ zero vector. Edges in $\Gamma_k$ correspond to recovery sets of size two. Therefore, two vertices of $\Gamma_k$ are connected if and only if the Hamming distance between their labels is one. The lemma immediately follows from this observation. (The bipartite structure of $\Gamma_3$ can be observed in \Cref{fig:CodeMatchGraph}(b) and \Cref{fig:SimplexMatch}.)
\end{proof}

We use the above lemma to characterize the rate region of Simplex codes in the following theorem.

\begin{theorem}
\label{thm:simplex-codes}
The service rate region of the $[2^{k}-1,k]$ Simplex coded system with $\mu = 1$ consists of all demand vectors $(\lambda_1,\cdots,\lambda_k)$ such that $\lambda_1+\lambda_2+\cdots+\lambda_k \leq 2^{k-1}$.
\end{theorem}

\begin{proof}
Consider first the non-overlapping recovery sets of size at most two. Let $\Gamma_k$ be the corresponding recovery graph. Consider an arbitrary demand vector $(\lambda_1,\cdots,\lambda_k)$ such that ${\lambda_1+\cdots+\lambda_k\leq 2^{k-1}}$. Assign weight $\lambda_i/2^{k-1}$ to each edge in $\Gamma_k$ that corresponds to a recovery group of object $i$. Note that this assignment forms a valid fractional matching (cf.\ Lemma~\ref{lem:simplex-bipartite}). 
It follows from Lemma~\ref{lem:simplex-bipartite} that the vertex cover number of $\Gamma_k$ is $2^{k-1}$, and thus, by Proposition~\ref{prop:vertex-cover}, no demand vector ${(\lambda_1,\cdots,\lambda_k)}$ can be served 
so that ${\lambda_1+\lambda_2+\cdots+\lambda_k > 2^{k-1}}$. 

Next, we show that larger (overlapping) recovery sets of size three do not increase the service rate region. To show this, let us add hyperedges to $\Gamma_k$ corresponding to recovery sets of size greater than two. Note that $2^{k-1}$ vertices with odd number of $1$'s also cover all the hyperedges. Indeed, if this is not the case, there must be a recovery set consisting of servers corresponding to vertices each having an even Hamming weight. Clearly, this is not possible, since the labels of vertices in a recovery set must add to a unit vector.
\end{proof}


It is worth noting two interesting observations. First, it well-known that the codewords of a $[2^k-1,k]$ Simplex code form a simplex in the Hamming space of binary length-$(2^k-1)$ vectors.
Interestingly, the service rate region of a $[n=2^k-1,k]$ Simplex code is a $(k-1)$-dimensional simplex in $\mathbb{R}^n$ defined as $\sum \lambda_i \leq 2^{k-1}$, $\lambda_i \geq 0$, $i\in [k]$.
As we will see in the next section, looking at codes through the lens of finite geometry enables us to characterize the service rate region of first-order Reed-Muller codes.

Second, from the achievability proof, one can observe that when a server completes a request, it simply starts serving the next request in the queue (if any). 
A natural question is how many users can be \textit{simultaneously} served by the Simplex-coded system in parallel? This is especially important for scenarios when each user occupies the entire bandwidth of the server. As we discuss in Sec.~\ref{sec:asynchronous-service-rate-region}, this question motivates us to introduce the notion of asynchronous service rate region. 


\subsubsection{MDS Codes}\label{Sec:recoverygraphMDS}
We show how a graph representation of MDS codes allows one to obtain bounds on the service rate region.
In a system using an $[n,k]$ systematic MDS code, a data object can be recovered from its systematic copy or from any $k$ of the remaining servers. In other words, the recovery sets of an MDS code are of size either one or $k$ (see \Cref{fig:CodeMatchGraph}(a) for the recovery graph of a systematic $[4,2]$ MDS code).

We represent an MDS code using a hypergraph. Note that a is a hypergraph any subset of vertices may be joined by an edged, referred to as a hyperedge, rather than a pair of vertices as in graphs (see, e.g.,~\cite[Chapter 7]{Voloshin:09:hypergraph}). Specifically, given an $[n,k]$ MDS code with a generator matrix $G$, we form the recovery hypergraph $\Gamma_G$  such that it contains a vertex for each column of $G$ and a hyperedge for every recovery set. We label every hyperedge of $\Gamma_G$ with the data symbol whose recovery set it is associated with. For each systematic column of $G$, we add $k-1$ additional vertices to $\Gamma_G$, and connect them with a hyperedge labeled with the corresponding symbol. 
As we see next, the recovery graph of an MDS code has a specific structure. The proof follows from the graph construction.

\begin{lemma}[Structure of the Recovery Graph for MDS Codes]
\label{lem:MDS-hypergraph}
For an $[n,k]$ MDS code with generator matrix $G$, the following holds.
\begin{enumerate}
    \item If $G$ has no systematic columns, then $\Gamma_G$ is a complete hypergraph on $n$ vertices with $k$ parallel hyperedges connecting every $k$-subset of vertices.
    \item If $G$ is systematic, then $\Gamma_G$ has $n+k(k-1)$ vertices with hyperedges of size $k$.
 \end{enumerate}
\end{lemma}

The above lemma allows us to obtain bounds on the service rate region of MDS codes as follows.\footnote{The bounds obtained here for systematic MDS codes are loose as compared to those obtained in Theorem~\ref{thm:converse_bnd} using the waterfilling algorithm in \Cref{sec:MDS}. Note that the waterfilling algorithm is defined only for systematic codes, whereas recovery (hyper)graphs can be used to analyze non-systematic codes as well.}
\begin{proposition}
\label{prop:rate-region-MDS-using-graphs}
For a system using an $[n,k]$ MDS code with no systematic nodes and $\mu=1$, the service rate region is the set of all request vectors $(\lambda_1, \cdots, \lambda_k)$ satisfying $\sum_{i=1}^{k}\lambda_i \leq n/k$. For a system using a systematic $[n,k]$ MDS code and $\mu=1$, the service rate region lies inside the region described by 
\mbox{$\sum_{\substack{i\in\mathcal{I,}\\ \mathcal{I}\subseteq\{1,\dots,k\}}}\lambda_i \leq k + \frac{|\mathcal{I}|}{k}(n-k)$.}
\end{proposition}
\begin{proof}
Consider an $[n,k]$ MDS code with no systematic symbols. For any $(\lambda_1,\cdots,\lambda_k)$ in its service rate region, there must exist a fractional matching in hypergraph $\Gamma_G$ such that the sum of the weights of hyperedges of label $i$ is $\lambda_i$ by \Cref{prop:fractional-matching}.
Obtain a hypergraph $\Gamma'_G$ from $\Gamma_G$ by collapsing parallel hyperedges connecting every $k$-subset of vertices into one hyperedge and assign its weight to be the sum of the weights of the parallel hyperedges. From Lemma~\ref{lem:MDS-hypergraph}, $\Gamma'_G$ is a complete hypergraph on $n$ vertices with each hyperedge of cardinality $k$ (which is known as a $k$-uniform hypergraph). Note that any fractional matching in $\Gamma_G$ induces a valid fractional matching in $\Gamma'_G$ since the sum of the weights on hyperedges incident to any vertex does not change.
Then, the converse and the achievability follow by noting that, for a complete $k$-uniform hypergraph on $n$ vertices, the fractional matching number is $n/k$ (see, e.g.,~\cite{Scheinerman:13:fractional-graph}).

The proof of upper bound for a systematic code essentially follows the same steps as above. Consider the case when, for some $\mathcal{I}\subseteq\{1,2,\dots,k\}$, $\lambda_i > 0$ for $i\in\mathcal{I}$ and $\lambda_i = 0$ otherwise. Then, by Prop.~\ref{prop:fractional-matching}, for $i\in\{1,2,\dots,k\}\setminus\mathcal{I}$, every hyperedge labeled with $\lambda_i$ must have zero weight in any fractional matching corresponding to a valid service allocation.
Let $\Gamma'_G$ be the graph obtained from $\Gamma_G$ by removing all hyperedges labeled $\lambda_i$ for each $i\in\{1,2,\dots,k\}\setminus\mathcal{I}$, and deleting the dummy vertices corresponding to the systematic columns $i\in\{1,2,\dots,k\}\setminus\mathcal{I}$. In other words, after removing the hyperedges labeled $\lambda_i$, we delete the resulting independent vertices.
Note that any fractional matching in $\Gamma_G$ corresponding to a valid service allocation is also a fractional matching in $\Gamma'_G$. 
Further, from Lemma~\ref{lem:MDS-hypergraph}, it is straightforward to see that the pruned graph $\Gamma'_G$ contains \mbox{$N = n + |\mathcal{I}|(k-1)$} vertices.
Therefore, the bound follows from Prop.~\ref{prop:fractional-matching} by using the fact that the fractional matching number for a $k$-uniform hypergraph on $N$ vertices is at most $N/k$ (see, e.g.,~\cite{Scheinerman:13:fractional-graph}).
\end{proof}

As an example, consider the non-systematic $[8,2]$ MDS code shown in the second row (in blue) in~\Cref{fig:MaxRegion}. The corresponding recovery graph is the complete graph on $8$ vertices. As shown in the proposition, the service rate region is the simplex $\lambda_1,\lambda_2\geq 0$, $\lambda_1+\lambda_2\leq 4$, depicted in blue in~\Cref{fig:MaxRegion}. 

For an example of a systematic MDS code, consider~\Cref{fig:CodeMatchGraph} (a). It shows a $[4,2]$ code along with the corresponding recovery graph having $6$ vertices. As per Proposition~\ref{prop:rate-region-MDS-using-graphs}, the sum of arrival rates $\lambda_1+\lambda_2 \leq 3$. Now, consider the case that only $a$ is being requested, i.e., $\lambda_2 = 0$. Then, we prune the recovery graph to remove the edges labeled with $b$ and deleting dummy vertices corresponding vertices. The pruned graph consists of the vertices $a$ and $0_a$ connected by an edge, and the triangle formed by $b$, $a+b$, and $a-b$. It is easy to see that the fractional matching number for this graph is at most $5/2$, and thus, $\lambda_1\leq 5/2$.

\subsection{Integral Service Rate Region and  Batch Codes
\label{sec:batch_codes_connection}}

Recall that in the bandwidth model (In Sec.~\ref{sec:bandwidth-model}), each server can concurrently serve only a limited number of data access requests. In this model, $\lambda_i$ represents the number of requests for object $i$ that are simultaneously present in the system. Even though this model results in the same allocation problem as in the queuing model (Sec.~\ref{sec:queuing-model}), it opens up interesting questions. In particular, consider scenarios wherein each user occupies the entire bandwidth of the server they are accessing. For example, this can happen when users are streaming from low-bandwidth edge devices and each user needs to be served at a specific rate. Motivated by these scenarios, we introduce the notion of integral service rate region defined as follows.

\begin{definition}
\label{def:integral-rate-region}
The integral service rate region of an $[n,k]$ coding scheme is a set of demand vectors $(\lambda_1,\cdots,\lambda_k)$ for which there exists a valid allocation $\{\lambda_{i,j}:1\leq i\leq k, 1\leq j\leq t_i\}$ satisfying~\eqref{eq:allocation-constraints_1}, \eqref{eq:allocation-constraints_2} and \eqref{eq:allocation-constraints_3} such that each $\lambda_{i,j}$ is an integer. 
\end{definition}

From the definition of the integral service rate region and Proposition~\ref{prop:fractional-matching} it follows that
a system using an $[n,k]$ code with a generator matrix $G$ contains a demand vector $(\lambda_1, \lambda_2, \cdots, \lambda_k)$ in its integral service rate region if and only if there exists an integral matching in the recovery graph $\Gamma_G$ such that the weights on the edges with label $i$ sum to $\lambda_i$.
An open problem associated with this observation is discussed in \Cref{sec:integralSRR}.


Next, we show that, if the integral service rate region of a coding scheme contains a specific region (in particular, all demand vectors $(\lambda_1, \cdots, \lambda_k)$ such that $\sum_{i=1}^{k}\leq t$ for a positive integer $t$), then the coding scheme must be a \textit{batch code}~\cite{Ishai:04:batch-codes} -- a well-known class of codes in computer science. 


Batch codes (in particular, multiset batch codes) are designed to simultaneously serve a certain number of requests (each asking for one object)
such that the worst-case maximal load on the system as well as the total amount of used storage are minimized~\cite{Ishai:04:batch-codes}. In the simplest form (called {\it primitive} batch codes), $k$ data objects are encoded into $n$ objects, which are distributed among $n$ servers (one object per server). The encoding should be such that an arbitrary subset (or batch) of $t$ objects can be decoded by simultaneous reads from a (sub)set of the servers.
The formal definition of primitive multiset batch codes is as follows (see~\cite{Ishai:04:batch-codes,Paterson:09,Skachek:18}).
\begin{definition}
\label{def:batch-codes} An $(n,k,t)$ (primitive multiset) batch code over $\mathbb{F}_q$ encodes $k$ objects $x_1, \cdots, x_k\in\mathbb{F}_q$ into $n$ objects $c_1,\cdots,c_n\in\mathbb{F}_q$ in such a way that for any multiset $i_1,\cdots,i_t\in[k]$, there is a partition of the servers into subsets $S_1,\cdots,S_t\subseteq[n]$ such that each object $x_{i_j}$, $j\in[t]$, can be recovered by downloading (at most) one object from each server in $S_j$. 
\end{definition}
Now, we show that the service rate region problem can be seen as a generalization of the batch code problem.
%
\begin{proposition}
\label{prop:batch-codes}
The integral service rate region of a storage system using an $[n,k]$ code with $\mu = 1$ includes all demand vectors $(\lambda_1,\cdots,\lambda_k)$ such that $\sum_{i=1}^k \lambda_i \leq t$ if and only if the code is an $(n,k,t)$ batch code.
\end{proposition}
\begin{proof}
Consider an arbitrary demand vector $(\lambda_1,\cdots,\lambda_k)$ such that each $\lambda_i$ is a non-negative integer and $\sum_{i=1}^{k}\lambda_i\leq t$. Note that any such demand vector can be considered as a multiset of size $\sum_{i=1}^{k}\lambda_i$ (which is at most $t$), where element $i$ is repeated $\lambda_i$ times. The proposition then follows from the definitions of the integral service rate region and the multiset primitive batch codes. 
\end{proof}

Batch codes are related to a class of codes designed for private information retrieval (PIR)~\cite{Yekhanin:10:PIR}. The key property of PIR codes~\cite{Fazeli:15:ISIT,Skachek:18} is that they have a number of disjoint recovery sets. Specifically, 
a binary $[n,k]$ code is called a $t$-server PIR code if for every $i \in [k]$, there exist a partition of the servers into subsets $S_1,\cdots,S_t\subseteq[n]$ such that the object $x_{i}$ can be recovered by downloading (at most) one object from each server in $S_j$, for all $j\in[t]$. 
%
The integral service rate region for PIR codes immediately follows from Proposition~\ref{prop:batch-codes} as follows.
The integral service rate region of storage system using an $[n,k]$ code with $\mu = 1$ includes all demand vectors $(t\cdot e_1,\dots, t\cdot e_k)$, $t\in \mathbb{N}$, if and only if the code is a t-server PIR code.


\subsection{Summary}
\label{sec:combinatorial-summary}
In this section, we proposed a graph representation to capture recovery sets of a linear code, and showed that the service rate allocation problem for a given linear code is equivalent to the fractional matching problem on the recovery graph associated with the code. This enabled us to characterize the service rate region for binary Simplex codes. A natural future direction is to analyze the service rate region for non-binary Simplex codes using the graph-based techniques. We also introduced the notion of integral service rate region, where allocations are constrained to be integers. We proved that the problem of characterizing an integral service rate region can be viewed as a generalization of the problem of designing primitive multiset batch codes. Exploring connections between the general batch codes and the problem of (integral and general) service rate region is an interesting future direction.

\section{Service Rate Region Using Geometric Approach}\label{sec:reed_muller}

Here, we look at the service rate problem through a geometric point of view. The problem of characterizing the service rate region of a code is a linear constrained optimization problem. The geometric approach is a powerful technique for addressing this problem that provides upper bounds on the sum of each subset of arrival rates in any demand vector $(\lambda_1, \cdots, \lambda_k)$ that can be served by a linear code. That is, using this approach, one can obtain a finite set of half-spaces (upper bounds) whose intersection encompasses the service rate region of a given linear storage scheme. In this section, we first discuss a geometric view of a storage scheme. We then give a brief description of the approach, and finally, using two examples, we explain how the service rate regions of the binary first order Reed-Muller codes and binary Simplex codes are obtained by the geometric technique. For a formal description and more details, see~\cite{ServiceGeometric:KazemiKS20}, where this approach is introduced.

\subsection{Geometric Description of Storage Schemes}

The projective space of dimension $k-1$ over $\mathbb{F}_q$, denoted as $\operatorname{PG}(k-1,q)$ is the set of $k$-tuples of elements of $\mathbb{F}_q$, not all zero, under the equivalence relation given by $(x_1,\ldots,x_k) \sim (\alpha x_1, \ldots, \alpha x_k)$, $\alpha\ne 0$, $\alpha\in\mathbb{F}_q$. 

Consider a generator matrix $G$ of a linear $[n,k]_q$ code. Columns of $G$ are vectors in $\mathbb{F}_q^k$. Since in the service rate region problem, we are only concerned with linear dependence of the columns of $G$, we consider $G$ as a geometric object in the projective space $\operatorname{PG}(k-1,q)$. Each column $\mathbf{g}_i$ of $G$ determines a point, and all columns that are scalar multiples of $\mathbf{g}_i$ are the same point in $\operatorname{PG}(k-1,q)$. Therefore, the columns of $G$ form a multiset of points in $\operatorname{PG}(k-1,q)$. 
We denote this $n$-multiset by $\mathcal{G}$ and say that it is induced by $G$~\cite{dodunekov1998codes}. 

Any $2$-dimensional subspace of  $\mathbb{F}_q^k$ determines a line in $\operatorname{PG}(k-1,q)$ and any $k-1$ dimensional subspace of $\mathbb{F}_q^k$ determines a hyperplane in $\operatorname{PG}(k-1,q)$. For details, see~\cite{tsfasman1995geometric,beutelspacher1998projective}. Note that since we are working over finite fields, hyperplanes consist of a finite number of points.

For example, consider the binary $[7,3]$ Simplex code. The columns of its generator matrix are the seven non-zero vectors of $\mathbb{F}_2^3$, and the seven points in the projective space ${\operatorname{PG}(2,2)}$. 
\Cref{fig:CodeMatchGraph}-b shows the corresponding $7$-multiset, known as the Fano plane. Since $k=3$, the $7$ lines of the ${\operatorname{PG}(2,2)}$ are also the hyperplanes of this $2$-dimensional projective space. 

Consider next the first storage scheme in \Cref{fig:MaxRegion} with four replicas of $a$ and four replicas of $b$. In the $8$-multiset induced by this scheme, there are only two different points $(1,0)$ and $(0,1)$ and the multiplicity of each is four. Finally, consider the last storage scheme in \Cref{fig:MaxRegion} with three replicas of $a$, three replicas of $b$ and two independent linear combinations $a+b$ and $a+\alpha b$. In the $8$-multiset induced by this scheme, the points are $(1,0)$ and $(0,1)$, each with multiplicity three and $(1,1)$ and $(1,\alpha)$  each with multiplicity one.


\subsection{Geometric Interpretation of the Service Rate Region Problem}

The following proposition plays a key role in deriving upper bounds on cumulative rates that can be simultaneously served by linear codes. In particular, it provides an upper bound on the sum of each subset of rates in any demand vector $(\lambda_1, \cdots, \lambda_k)$ in the service rate region of the system. In other words, it shows that the service rate region lies inside a region defined as the intersection of a finite number of half spaces (derived upper bounds).

\begin{proposition}\label{coro_hyperplane_constraint}
For a system using an $[n,k]$ code with $n$-multiset ${\mathcal{G}}$ in $\operatorname{PG}(k-1,q)$, consider any arbitrary demand vector $(\lambda_1, \cdots, \lambda_k)$ in its service rate region. Then, it holds that
\[
  \sum_{\substack{i\in\mathcal{I}\\ \mathcal{I}\subseteq\{1,\dots,k\}}}\lambda_i \leq {\mu\cdot|\mathcal{G} \setminus \mathcal{H}|}
\]
where $\mathcal{H}$ is a hyperplane of $\operatorname{PG}(k-1,q)$ not containing $\mathbf{e}_i$ for all $i\in\mathcal{I}$. 
\end{proposition}

\begin{proof}

If a vector $(\lambda_1,\ldots,\lambda_k)$ is in the service rate region, then it holds that simultaneously the cumulative request rate of $\sum_{\substack{i\in\mathcal{I}\\ \mathcal{I}\subseteq\{1,\dots,k\}}}\lambda_i$ for all objects ${i \in \mathcal{I}}$ is served by the system. On the other hand, since the hyperplane $\mathcal{H}$ does not contain any unit vector ${\mathbf{e}_i}$ for all ${i \in \mathcal{I}}$, the points contained in $\mathcal{H}\cap\mathcal{G}$, on their own, are not able to generate $\mathbf{e}_i$ for all ${i \in \mathcal{I}}$. Thus, for each ${i \in \mathcal{I}}$, whatever the used recovery sets for object $i$ are, in order to serve the request rate $\lambda_i$ for object $i$, some points outside of $\mathcal{H}$ with the cumulative service rate of (at least) $\lambda_i$ must be used. Thus, in order to satisfy the cumulative request rate of $\sum_{\substack{i\in\mathcal{I}\\ \mathcal{I}\subseteq\{1,\dots,k\}}}\lambda_i$ for all objects ${i \in \mathcal{I}}$, the cumulative service rate of the points in $\mathcal{G}$ that are outside of $\mathcal{H}$ must be at least $\sum_{\substack{i\in\mathcal{I}\\ \mathcal{I}\subseteq\{1,\dots,k\}}}\lambda_i$. 
\end{proof}

\subsection{Using Geometric Approach to Characterize Service Rate Regions}
The geometric approach equipped with Proposition~\ref{coro_hyperplane_constraint} can be used to obtain the service rate region of Simplex and Reed-Muller codes. The service rate region of the binary ${[2^k-1,k]}$ Simplex code and binary $[2^{k-1},k]$ first order Reed-Muller code are given in~\cite{ServiceGeometric:KazemiKS20}.
To illustrate the method, we here provide two examples: binary $[7,3]$ Simplex code and binary $[8,4]$ first order Reed-Muller code.

\subsubsection{Simplex Codes}
\label{sec:geometric-simplex-codes} In this section, as an example, we show how the service rate region of the $[7,3]$ Simplex code is characterized using the geometric approach. Consider a storage system using the $[7,3]$ Simplex code. Without loss of generality, assume that $\mu = 1$. Let $(x_1,x_2,x_3)$ denote a generic non-zero vector in $\mathbb{F}^3_2.$  Observe that the hyperplane $\mathcal{H}$ given by $\sum_{i=1}^3 x_i=0$ (namely, the hyperplane containing the points $(0,1,1)$, $(1,0,1)$ and $(1,1,0)$ in the Fano plane depicted in~\Cref{fig:CodeMatchGraph}-b) does not contain any unit vector $\mathbf{e}_i$, $i \in \{1,2,3\}$. Thus, for any demand vector ${\boldsymbol{\lambda}=(\lambda_a,\lambda_b,\lambda_c)}$ in the service rate region, applying Proposition~\ref{coro_hyperplane_constraint} results in $\lambda_a+\lambda_b+\lambda_c\le 4$. The reason is that the hyperplane $\mathcal{H}$ does not contain the points $(0,0,1)$, $(0,1,0)$, $(1,0,0)$ and $(1,1,1)$. Thus, so far we have shown that the service rate region is contained in the polytope $\mathcal{P}=\left\{\boldsymbol{\lambda}\in\mathbb{R}^3_{\ge 0}\,:\, \sum_{i=1}^3 \lambda_i\le 4\right\}$.
It is interesting to note that, in the recovery graph, that the vertices corresponding to the points outside of $\mathcal{H}$ form a minimal vertex cover (see~\Cref{fig:CodeMatchGraph}-b).

For the achievability proof, since the service rate region is a convex subset of $\mathbb{R}^3_{\geq 0}$, we only need to show that the vertices of the polytope $\mathcal{P}$, i.e., $(0,0,0)$, $(4,0,0)$, $(0,4,0)$ and $(0,0,4)$, are in the service rate region of this storage system. To see that, we observe that there are four disjoint recovery sets for each data object (see \Cref{fig:SimplexStore} for object $a$), and thus the request rate of $1$ can be assigned to each of these recovery sets without violating the node capacity constraints. 

\subsubsection{First-Order Reed-Muller Codes}
\label{sec:geometric-RM-codes}
To show a sketch of the proof in characterizing the service rate region of the first order Reed-Muller code, we consider a non-systematic $[8,4]$ first order Reed-Muller code. Consider a system where four objects $a$, $b$, $c$, and $d$ are stored across $8$ servers using the first order Reed-Muller code $\text{RM}_2(1,3)$ with a non-systematic generator matrix as follows:
\[
{G}= 
    \begin{bmatrix}
     1 & 1 & 1 & 1 & 0 & 0 & 0 & 0 \cr
      1 & 1 & 0 & 0 & 1 & 1 & 0 & 0 \cr
     1 & 0 & 1 & 0 & 1 & 0 & 1 & 0 \cr
     1 & 1 & 1 & 1 & 1 & 1 & 1 & 1 \cr
    \end{bmatrix}
\]
which encodes ${[a, b, c, d]}$ into $[a+b+c+d, a+b+d, a+c+d, a+d, b+c+d, b+d, c+d, d]$. The recovery sets for object $a$ are shown in~\Cref{fig:ReedMullerRecoverya}.
\begin{figure}[hbt]
    \centering
    \includegraphics[scale=0.25]{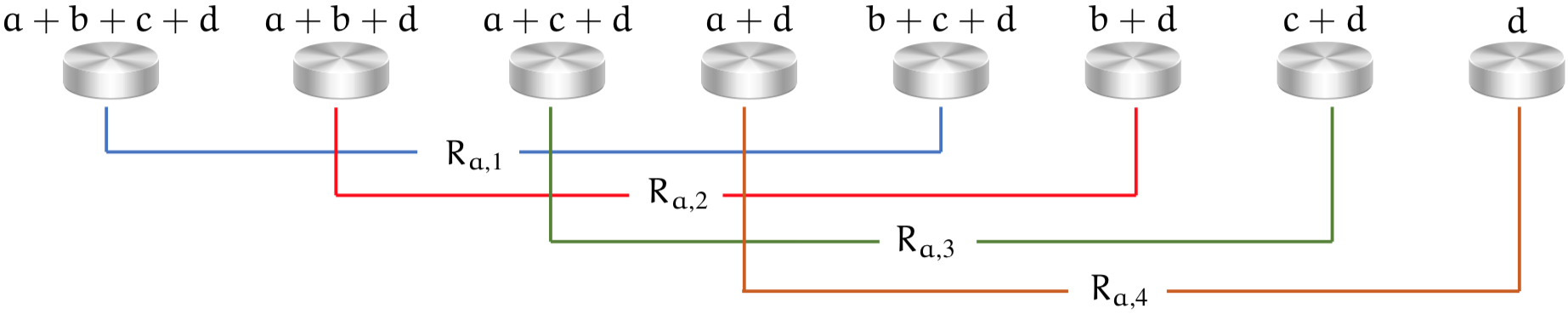}
    \caption{Recovery sets for data object $a$ in the $[8,4]$ Reed-Muller code.}
    \label{fig:ReedMullerRecoverya}
\end{figure}
The recovery sets for objects $b$ and $c$ can be obtained similarly to those for $a$. The recovery sets for object $d$ are shown in~\Cref{fig:ReedMullerRecoveryd}. 
\begin{figure}[hbt]
    \centering
    \includegraphics[scale=0.25]{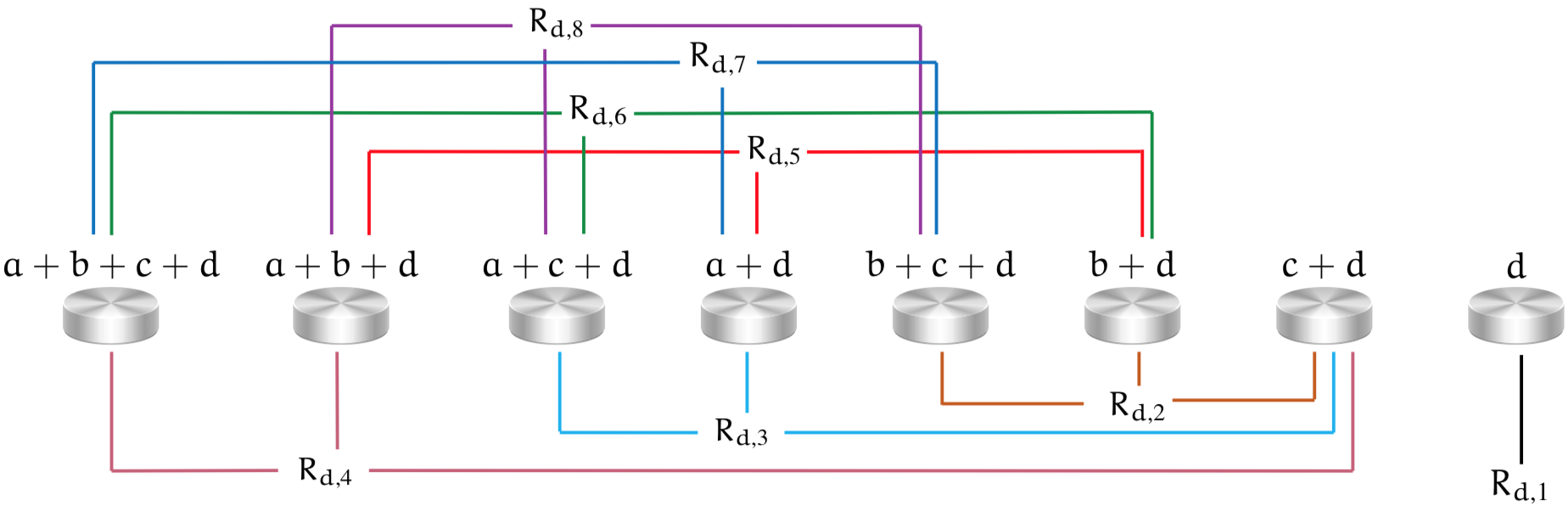}
    \caption{Recovery sets for data object $d$ in the $[8,4]$ Reed-Muller code.}
    \label{fig:ReedMullerRecoveryd}
\end{figure}

Let $(x_1,\dots,x_4)$ be a non-zero vector in $\mathbb{F}^4_2$. Observe that the hyperplane $\mathcal{H}$ given by $\sum_{i=1}^4 x_i=0$ does not contain any unit vector $\mathbf{e}_i$, $i \in [4]$. The hyperplane $\mathcal{H}$ does not contain the $4$ column vectors $(1,1,0,1)$, $(1,0,1,1)$, $(0,1,1,1)$ and $(0,0,0,1)$ of the generator matrix. Thus, for any demand vector ${\boldsymbol{\lambda}=(\lambda_a,\lambda_b,\lambda_c,\lambda_d)}$ in the service rate region, applying the Proposition~\ref{coro_hyperplane_constraint} results in the constraint below
\begin{equation}
  \label{ie_upper_1}
  \lambda_a + \lambda_b + \lambda_c + \lambda_d \leq 4.
\end{equation}

On the other hand, the unit vector $\mathbf{e}_i$ for all ${i \in \{1,2,3\}}$ is not a column of the generator matrix which means that files $a$, $b$, and $c$ do not have any systematic recovery sets. Thus, for files $a$, $b$, and $c$, the cardinality of all recovery sets is at least two, and the minimum system capacity utilized by ${\lambda_i}$ for ${i \in \{a,b,c\}}$ is ${2\lambda_i}$. For file $d$, since all columns of the generator matrix have one in the last row, the cardinality of every recovery set is odd. Hence, for file $d$, the unit vector $\mathbf{e}_4$, which is a column of $G$, forms a recovery set of cardinality one, while all other recovery sets have cardinality at least three. Thus, the minimum system capacity utilized by ${\lambda_d}$ for ${\lambda_d \leq 1}$ is ${\lambda_d}$ and for ${\lambda_d \geq 1}$ is ${1+3(\lambda_d-1)}=3\lambda_d-2$. Since the system has $8$ servers, each of service capacity $1$, based on the capacity constraints, the total capacity utilized by the requests for download must be at most $8$. Thus, any vector $(\lambda_a,\lambda_b,\lambda_c,\lambda_d)$ in the service rate region must satisfy the following:
\begin{equation}  \label{ie_upper_2}
 \begin{cases}
   2(\lambda_a + \lambda_b + \lambda_c) + \lambda_d \leq 8 & ~\text{for} ~~ \lambda_d \leq 1 \\
  2(\lambda_a + \lambda_b + \lambda_c) + 3 \lambda_d - 2 \leq 8 & ~\text{for} ~~ \lambda_d \geq 1
  \end{cases}
\end{equation}

We showed that the service rate region lies inside the polytope $\mathcal{T}=\left\{\boldsymbol{\lambda}\in\mathbb{R}^4_{\ge 0}\,:\, \boldsymbol{\lambda} ~\text{satisfies}~ \eqref{ie_upper_1}, \eqref{ie_upper_2}\right\}$. Suppose ${\lambda_b=\lambda_c=0}$. 
\Cref{sr} depicts the service rate region of this storage scheme in the ${\lambda_a-\lambda_d}$ plane wherein \eqref{ie_upper_1} and~\eqref{ie_upper_2} are respectively shown with the red line and the green line. 

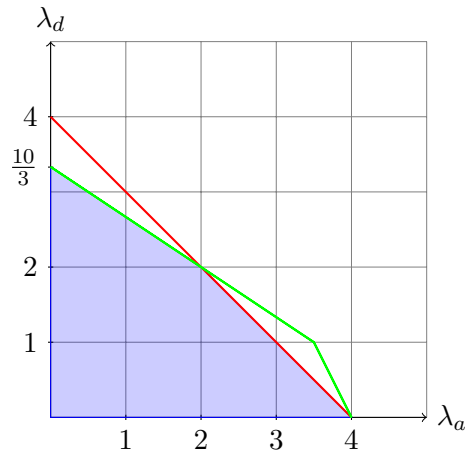
\begin{figure}[hbt]
    \centering
     \begin{tikzpicture}[domain=0:5] 
  \draw[very thin,color=gray] (0,0) grid (5,5);
    \draw[->] (0,0) -- (5,0) node[right] {$\lambda_a$}; 
    \draw[->] (0,0) -- (0,5) node[above] {$\lambda_d$};
    \foreach \x/\xtext in {1, 2, 3, 4} 
   \draw (\x cm,1pt) -- (\x cm,-1pt) node[anchor=north,fill=white] {$\xtext$};
 \foreach \y/\ytext in {1,2, 3.33/\frac{10}{3}, 4} 
   \draw (1pt,\y cm) -- (-1pt,\y cm) node[anchor=east,fill=white] {$\ytext$};
    \filldraw[draw=blue,fill=blue,fill opacity=0.2] (0,10/3) -- (2,2) -- (4,0) -- (0,0) -- (0,10/3);
    \filldraw[thick,draw=red,fill=red,fill opacity=0.2] (0,4) -- (4,0);
    \filldraw[thick,draw=green,fill=green,fill opacity=0.2] (0,10/3) -- (3.5,1) -- (4,0) -- (3.5,1) -- (0,10/3);
\end{tikzpicture}  
    \caption{Service rate region of the $[8,4]_2$ first order Reed-Muller code in ${\lambda_a-\lambda_d}$ plane with ${\lambda_b=\lambda_c=0}$ where the equation constraints~\eqref{ie_upper_1} and~\eqref{ie_upper_2} are respectively shown with the red line 
    and the green line. 
    }
    \label{sr}
\end{figure}

For the achievability proof, one only needs to provide constructions for the vertices of the polytope $\mathcal{T}$ in ${\lambda_a-\lambda_d}$ plane. The demand vector ${(\lambda_a,\lambda_b,\lambda_c,\lambda_d)=(4,0,0,0)}$ can be achieved by assigning the request rate of $1$ to each of the $4$ disjoint recovery sets of file $a$ shown in~\Cref{fig:ReedMullerRecoverya}. For the  $(\lambda_a,\lambda_b,\lambda_c,\lambda_d)=(2,0,0,2)$, the ${\lambda_a=2}$ can be served by assigning the request rate of $1$ to each of the recovery sets $(b+d,a+b+d)$ and $(b+c+d,a+b+c+d)$, and $\lambda_d=2$ can be satisfied by assigning the request rate of $1$ to the systematic recovery set $(d)$, and the request rate $1$ to the recovery set $(a+d,c+d,a+c+d)$ of file $d$. For the demand vector $(\lambda_a,\lambda_b,\lambda_c,\lambda_d)=(0,0,0,\frac{10}{3})$, the ${\lambda_d=\frac{10}{3}}$ can be served without violating the node capacity constraints by assigning the request rate of $1$ to the systematic recovery set $(d)$, and the request rate of $\frac{1}{3}$ to each of the $7$ recovery sets of size $3$ for file $d$, depicted in~\Cref{fig:ReedMullerRecoveryd}.

\subsection{Summary}\label{sec:geometric-summary}
In this section, we proposed a geometric technique for addressing the problem of characterizing the service rate region of a given linear storage scheme without explicitly listing the set of all possible recovery sets. By leveraging the proposed geometric technique, initial steps were taken towards deriving upper bounds on the service rate regions of some parametric classes of linear codes. In particular, upper bounds on the service rate regions of the binary first order Reed-Muller codes and binary simplex codes, as two classes of codes which are important in both theory and practice, were derived. Then, it has been shown that how the derived upper bounds can be achieved. Utilizing the geometric technique to investigate the service rate regions of other common coding schemes such as MDS codes, second order Reed-Muller codes, non-binary Reed-Muller codes, and non-binary simplex codes are amongst the most natural future directions.

\section{Ongoing and Open Problems}
\label{sec:open_problems}

This paper presented several initial fundamental results and techniques concerning the service rate region of a distributed coded storage system. We summarised the findings at the end of each section. \Cref{tab:summary-of-codes} outlines the results concerning particular code classes. The table also indicates the limitations of this early work. Thus, many direct extensions and continuations within the described thrusts are apparent. Some compelling problems of varying degrees of difficulty include, e.g., extending the results to other classes of codes. 


There are many related problems just outside of the main scope of the paper. This section presents a summary of connected ongoing and open problems along two threads: 1) performance analysis of storage schemes, which requires queueing and combinatorial optimization expertise, and 2) designing the storage schemes to maximize the service rate region, which requires information and coding theory expertise. Since each of these problems would greatly benefit from jointly solving both the performance analysis and code design problems, we believe that these directions will bridge deeper connections between these communities.

\subsection{Performance Analysis and Networking Problems}

\subsubsection{The Coverage of a Rate Region}

The service rate region of a given storage scheme \emph{covers} the set of achievable demand vectors $\boldsymbol{\lambda} = (\lambda_1, \lambda_2, \dots, \lambda_k)$. Since this is a multi-dimensional region, we need a way to map it to a single scalar metric that can be used to objectively compare two service rate regions. One candidate is the volume of the $k$-dimensional region. Instead, we propose a more natural candidate metric -- we consider a probability distribution $f_{\boldsymbol{\Lambda}}(\boldsymbol{\lambda})$ of demand vectors and measure the fraction of requests that are \emph{covered} by the service rate region. The coverage or the covered mass of a rate region $\mathcal{S}$ is defined as
\begin{align}
M(\mathcal{S}, f_{\boldsymbol{\Lambda}}) = \int_{\boldsymbol{\lambda} \in \mathcal{S}}   f_{\boldsymbol{\Lambda}}(\boldsymbol{\lambda}) d \boldsymbol{\lambda}. \label{eqn:coverage}
\end{align}

For example, in \Cref{fig:covered_fraction} we compare the coverage of the replication and MDS coding storage schemes for $k=2$ objects stored on $n=4$ servers. The heatmap shows the probability distribution $f_{\boldsymbol{\Lambda}}(\boldsymbol{\lambda})$ of the demand vectors, where one of $a$ or $b$ is likely to be in high demand, but both objects are not in high demand simultaneously. The MDS coded system has a coverage of 
$M_{\text{MDS}} = 0.91$, which is larger than the coverage $M_{\text{Rep}} = 0.82$ of the replicated system. 

\begin{figure}[htb]
 \centering
\begin{tikzpicture}
\node at (-4,0) { \includegraphics[scale=0.75]{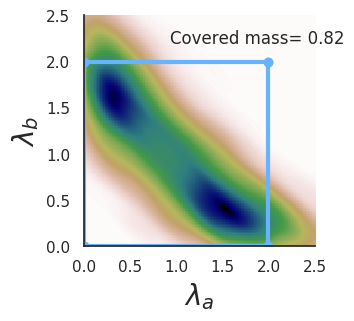}};
\node at (4,0) { \includegraphics[scale=0.75]{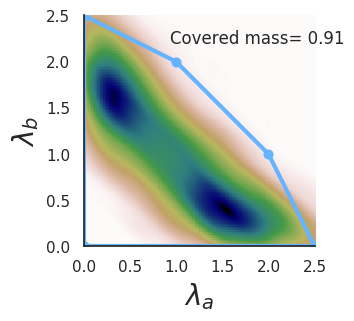}};
\end{tikzpicture}
  \caption{The demand rate distribution shown as a heatmap for (left) replication system $(a,a,b,b)$ and (right) MDS coded system $(a,b,a+b,a+2b)$. Blue lines show the boundary of system's service rate region. Covered mass is the cumulative probability mass for the demand vectors that lie within the system's service capacity region.
 }
  \label{fig:covered_fraction}
\end{figure}
Analyzing the service rate regions of well-known classes of codes for typical demand or content popularity distributions such as the Zipf distribution is an open performance analysis problem. It will provide valuable insights that can be used in designing codes that provide maximum coverage with the minimum number of nodes.


\subsubsection{Analyzing the Cost of Serving Requests}

Serving a download request collaboratively by two or more servers (some of which store encoded objects) occupies more system resources than serving it at a single server. For instance, accessing $a$ from $b$ and $a+b$ requires downloading two objects to access one object. 
An interesting research direction is to study this cost quantitatively. In the following we formally propose the cost associated with a given storage scheme. 

Let us define the \textit{service cost} of a single request as the number of objects that are downloaded in order to satisfy the request, which is the size of the corresponding recovery group $R_j$. We define the normalized service cost $C(\boldsymbol{\lambda})$ of a demand vector $\boldsymbol{\lambda}$ as the cumulative transfer rate required by the servers to serve this demand, divided by the sum of the elements of $\boldsymbol{\lambda}$, that is,
\begin{align}
C(\boldsymbol{\lambda}) = \frac{\sum_{i=1}^{k} \sum_{j=1}^{t_i} |R_j| \cdot\lambda_{i,j}}{\sum_{i=1}^{k} \sum_{j=1}^{t_i} \lambda_{i,j}} ,
\end{align}
where $\lambda_{i,j}$ is the portion of the request rate $\lambda_i$ allocated to the $j^{th}$ repair group $R_j$, whose size is $|R_j|$. The service cost $C(\boldsymbol{\lambda})$ represents the amount of data that is downloaded per request, and it depends on the underlying coding scheme and the request allocation scheme used to split requests across recovery groups. For example, if we use replication coding, then the size of each recovery group $|R_j| = 1$, and thus, $C(\boldsymbol{\lambda}) = 1$. With erasure coding, a request may need to download two or more coded objects to recover one data object, and thus we will incur a higher cost. 
%
%
%
%
Consider again the example of serving the demand with $\lambda_a=1.5\mu$ and $\lambda_b=0.5\mu$ would cost $2\mu$ in the replicated system $(a, a, b, b)$ and $4\mu$ in the MDS coded system $(a, b, a+b, a+2b)$. Normalizing these by the demands, we get the cost $C(\boldsymbol{\lambda}) =1$ for the replicated system and $C(\boldsymbol{\lambda}) =1.6$ for the MDS-coded system. \Cref{fig:fig_coveredmass_Ecost} shows a heat map of the \textit{normalized service cost} for all demand vectors within the service rate region of two systems. 


\begin{figure}[ht]
  \centering
 \begin{tikzpicture}
  \node at (-4,0) {\includegraphics[scale=0.55]{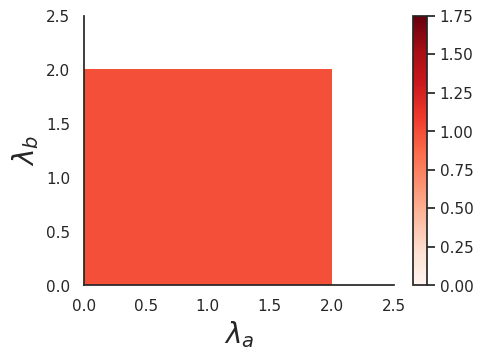}};
  \node at (4,0) {\includegraphics[scale=0.55]{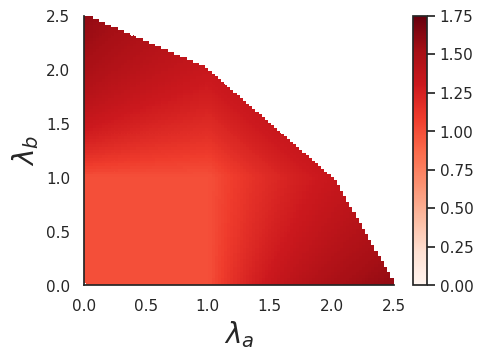}};
 \end{tikzpicture}
  \caption{The normalized service cost $C(\boldsymbol{\lambda})$ shown as a heat map for (left) replicated system $(a,a,b,b)$ and (right) MDS system 
  $(a,b,a+b,a+2b)$. The normalized service cost at each point in the service rate region is defined as the number of bits that need to be downloaded per data bit. For the replicated system, the cost is $C(\boldsymbol{\lambda}) = 1$, whereas for the MDS system that cost varies from $1$ to $1.6$ depending on the skewness of the object demands $(\lambda_a, \lambda_b)$.}
  \label{fig:fig_coveredmass_Ecost}
\end{figure}

An open question for future research is to compare the expected costs of different service rate regions. We can again consider a probability distribution $f_{\boldsymbol{\Lambda}}(\boldsymbol{\lambda})$ of demand vectors and measure the fraction of requests that are \emph{covered} by the service rate region. Then the expected service cost of a rate region $\mathcal{S}$ is $
C(\mathcal{S}, f_{\boldsymbol{\Lambda}}) = \int_{\boldsymbol{\lambda} \in \mathcal{S}}   C(\boldsymbol{\lambda}) f_{\boldsymbol{\Lambda}}(\boldsymbol{\lambda}) d \boldsymbol{\lambda}$. We conjecture that for imbalanced demand distributions, where several objects are not in high demand at the same time, the coded systems will incur little additional cost, but will have higher coverage.

Instead of measuring the service cost in terms of the amount of data accessed to serve one request as captured by $C(\boldsymbol{\lambda})$, an alternate metric is to consider the total computing time spent serving each request, as considered in \cite{DownloadTimeOfAvailabilityCodes:AktasKS20,joshi2017efficient, joshi2015queues, joshi2018synergy, aktas2019straggler}. These papers identify regimes where coded data access or computing incurs a lower total computing time than replicated systems.

\subsubsection{Latency Analysis of Coded Systems}
In this paper, we focused on the service rate region, that is, the demand vectors $\boldsymbol{\lambda} = (\lambda_1, \lambda_2, \dots, \lambda_k)$ that can supported by the system while ensuring that the total request rate assigned to each server does not exceed its capacity $\mu$. We did not consider the delay experienced by each request.
%
The request splitting policies proposed in this paper such as the water-filling algorithm are throughput-optimal, and not necessarily delay-optimal. Open problems for future research include 1) analyzing the latency experienced by a user accessing data from coded distributed system and 2) designing delay-optimal policies for splitting requests across recovery groups, which is much harder than designing throughput-optimal policies. 

The first problem of analyzing the latency of accessing content from coded distributed systems, specifically MDS and availability coded systems, has been previously considered in \cite{joshi2012coding, joshi2014delay, shah2016when, swanand_isit_2015, joshi2017efficient, MDS:AktasS18}. However, these works consider redundant requests, that is, each request is sent to multiple recovery groups and the request is considered served when the data is successfully accessed from any one of the groups. They show that the resulting system is a generalized fork-join queueing system, whose analysis is a famously hard problem in queueing theory \cite{flatto1984two, nelson_tantawi, varki_merc_chen}, and find bounds on the expected latency. Besides the fork-join queue model, several approaches like block-$t$ and probabilistic scheduling have been developed to remove some of the key assumptions in latency analysis \cite{WhenQingMeetsCoding:ChenSH14, JointLatencyCost:XiangLA15}.
Analyzing the latency of hybrid systems that use a combination of replication and erasure coding, both with and without sending redundant requests to multiple recovery groups, as proposed in this paper is a pertinent open problem.

The second problem of designing delay-optimal policies, a highly challenging problem, requires striking the perfect balance between using several recovery groups in parallel to serve a set of incoming requests and queueing more requests for sequential processing at smaller recovery groups. Our current throughput-optimal allocation policies such as water-filling are biased towards the latter strategy because they give higher priority to sending requests to smaller recovery groups. An alternative approach is to design policies that perform the best possible \emph{load balancing} of requests across the nodes, that is, minimize the cumulative request rate assigned to the maximally loaded node, as recently considered in \cite{LoadBalancing:AktasBS19}. Such load-balancing strategies can give better latency performance than the throughput-optimal strategies considered in this paper.

\subsubsection{Expanding the Service Rate Region via Redundant Requests}
In this paper we assume that each data access request is sent to exactly one of its recovery groups. Instead, assigning requests to more than one recovery groups, waiting for any one copy to be downloaded and canceling the outstanding requests can reduce the latency experienced by that request. But too much redundancy can increase the waiting time in queue for subsequent requests. Several recent works study queueing systems with redundancy \cite{righter_job_rep_2009, gardner2015reducing, gardner2016power, gardner2016s&x, gardner2019little, joshi2014delay, joshi2015queues, joshi2017efficient, raaijmakers2019delta, raaijmakers2019redundancy, MDS:AktasS18}. Some of these works \cite{righter_job_rep_2009, joshi2017efficient, shah2016when} observe that when the distribution of the service time of each request is heavier than an exponential distribution (in particular, new-shorter-than-used distributions such as the hyper-exponential distribution) then redundancy expands the achievable rate region beyond the sum capacity of individual servers. That is, a system of $n$ servers with capacity $\mu$ requests per time each can support a demand higher than $n \mu$. This non-intuitive phenomenon has been recently studied in the context of replication of jobs in computing systems in \cite{joshi2018synergy, anton2019stability, anton2020improving}. However, understanding the expansion of the service rate region due to heavy-tailed and new-shorter-than-used service times for erasure coded distributed systems such as those considered in this paper is an open question. The effect of heavy-tailed distributions on the latency of jobs with many parallel tasks, where straggling tasks are replicated is also previously studied in \cite{wang2019efficient, wang_dcc_2015, wang_efficient_2014, aktas2019straggler}, albeit without considering queueing of tasks.

\subsection{Coding Theory and Data Allocation Problems}

Next we discuss some open problems from the perspective of code design or data allocation to achieve the best possible service rate region with minimum number of storage nodes.

\subsubsection{Designing Codes to Achieve a Rate Region with Desired Properties} 
In \Cref{sec:code_design_max_rate_region}, we introduced the problems of designing the coded storage schemes to maximize the volume of the service rate region with a given number of servers or cover a given service rate region with the minimum number of servers. Note that the service rate regions of two generator matrices ${G}$ and ${G'}$ of the same linear code might not be the same. Depending on the application, one may be interested in using a particular code with some desired properties. Then, these problems can be interpreted as finding the best generator matrix of a code with respect to the service rate region. Also, depending on the application, a metric different from the volume of the region might be of interest to a system designer. An alternative metric of practical interest is to optimize for a given number of servers $n$ and demand distribution $f_{\boldsymbol{\Lambda}}$ is the  fraction $M(\mathcal{S}, f_{\boldsymbol{\Lambda}})$ of the demand distribution $f_{\boldsymbol{\Lambda}}$ covered by the region, as defined in \eqref{eqn:coverage}. Using this metric can present some interesting challenges in designing the coding schemes. For example, consider two objects $a$ and $b$ whose demands are negatively correlated, that is, when $a$ is popular $b$ is not and vice-versa. Then encoding $a$ and $b$ together to form the coded object $a + \alpha b$ will give better coverage than encoding $a + \alpha c$, where $c$'s demand is positively correlated with $a$.

Another design objective can be to achieve the tail latency $\Pr(T > \tau) \leq \delta$ for a given demand distribution $f_{\boldsymbol{\Lambda}}$, deadline $\tau$ and tail probability $\delta$ with the minimum number of servers. For commonly observed demand distributions in which only a few objects are in high demand simultaneously, we expect coded storage schemes to significantly outperform uncoded and replicated storage schemes. 

\subsubsection{Data Striping Across Multiple Nodes and Multiple Objects Per Node} For simplicity in introducing the new concept service rate region, in this paper, we treat each data object as an atomic unit such that the $k$ data objects are used as $k$ information symbols when creating encoded versions of the objects. We also assume that each of the $n$ servers has the capacity to store exactly one data object. Distributed storage systems often employ data striping or sub-packetization \cite{raid_1988}, that is, dividing object $a$ is divided into stripes $a_1$, $a_2$, \dots, $a_s$, which are used as source symbols, and encoded and stored across different nodes. Spreading an object across more nodes can allow faster parallel reads of large objects. 

However, there are two possible drawbacks that come with parallelism. First, the impact of parallelism on the service rate region is not straightforward. Request service times might possibly go down super linearly with the reduced data size. That is, if downloading $a$ takes $t$ seconds, downloading each $a_i$ where $i = 1, \dots, s$ might take longer than $t/s$. In this case, parallelism might lead to smaller service rate region. Second, parallel download requires accessing multiple nodes simultaneously. Failure or slowdown at any one of these nodes can bottleneck the data access and increase latency. Removing the assumption of atomicity of each data object and generalizing the concept of service rate region to characterize the rate region of distributed storage with data striping is an open problem for future research. 

Even without data striping, each node may store multiple data objects, unlike our assumption that each server has the capacity to store exactly one data object and the entire object is accessed by each request. In this setting, designing optimal storage of objects so as to maximize the service rate region is an interesting open problem. In particular, as observed in \cite{LoadBalancing:AktasBS19}, minimizing the overlap between recovery groups of different objects could lead to a storage allocation that maximizes the service rate region. One possibility is to use expander graphs to design such storage allocation schemes.

Sub-packetization can also possibly make the system's service rate region larger. Whether it actually gets larger is not obvious. The answer depends on the scaling of the request service times with the object sizes. We elaborate on this connection below with an example. It is worth to note here that, in real storage systems (e.g., Google file system), object sizes are experimentally tuned to a value that is not too small in order to keep the system level overheads small compared to the actual time spent while fetching data from the nodes.

Recall that we define the node capacity as the maximum number of requests that can be served by a node per second. Suppose now we divide each object into two equal chunks and store them across separate servers. An object request will now be served by splitting it into two chunk requests and then assigning them to the respective servers. Obviously, the time to download a single chunk from a server, $t_1$, will be less than the time to download an object (two chunks), $t_2$. Note that, due to the system level overheads, download time is not a linear function of the data size. That is why, we cannot conclude that $t_1 = t_2 / 2$. Note also that, to download an object, the system now needs to serve two chunk requests instead of one object request. This overall means that downloading objects by fetching chunks from multiple servers can lead to consuming more system capacity than downloading the whole object from a single server. It is therefore not clear whether dividing objects into chunks can lead to larger service rate region.

\subsubsection{Integral Service Rate Region\label{sec:integralSRR}} 

We have defined the integral service rate region as a set of demand vectors $(\lambda_1,\cdots,\lambda_k)$ for which there exists a valid allocation\footnote{An allocation is valid if it satisfies 
\eqref{eq:allocation-constraints_1}, \eqref{eq:allocation-constraints_2} and \eqref{eq:allocation-constraints_3}.}
$\{\lambda_{i,j}:1\leq i\leq k, 1\leq j\leq t_i\}$ 
such that each $\lambda_{i,j}$ is an  integer (see Sec.~\ref{sec:batch_codes_connection}). Recall that $\lambda_i=\sum_{1\leq j\leq t_i} \lambda_{ij}$ for $1\leq i\leq k$, and thus, all points $(\lambda_1,\cdots,\lambda_k)$ in the integral service rate region have all integer components. An important open problem asks whether  
an integer component point $(\lambda_1,\cdots,\lambda_k)$ in the service rate region is always in the integral service rate region, i.e.,
whether for an integer component point $(\lambda_1,\cdots,\lambda_k)$ in the service rate region, there always exists  a valid allocation $\{\lambda_{i,j}:1\leq i\leq k, 1\leq j\leq t_i\}$ such that each $\lambda_{i,j}$ is a non-negative integer.
Confirming that this is true tells us (cf.~\Cref{prop:batch-codes}) that an $[n,k]$ code whose service rate region includes the points satisfying $\sum_{i=1}^k \lambda_i \leq t$ for some positive integer $t$ is an $(n,k,t)$ batch code. Thus the service rate region problem can be seen as a generalization of the batch code problem. 
The potentially more general service rate region problem may be easier to solve than the corresponding batch code problem. For example, proving that the binary ${[2^{k}-1,k]}$ Simplex code is a ${(2^{k}-1,k,2^{k-1})}$ batch code is fairly involved \cite{wang2017switch}, while deriving its service rate region is very straightforward by using either the combinatorial or the geometric techniques, as we did above. Another unexplored direction is using the techniques proposed in this paper to derive the batch properties of binary Hamming and Reed-Muller codes, previously considered in \cite{baumbaugh2017batch}).

We next consider an example to illustrate the integral service rate region problem we just stated as well as to point out a question concerning matching in graphs that, to the best of our knowledge, has not been asked before.
Consider the $[7,3]$ Simplex code  and its recovery graph as shown in~\Cref{fig:SimplexMatch}-(a) 
\label{ex:integral-service}
\begin{figure}[hbt]
    \centering
    \includegraphics[scale=0.99]{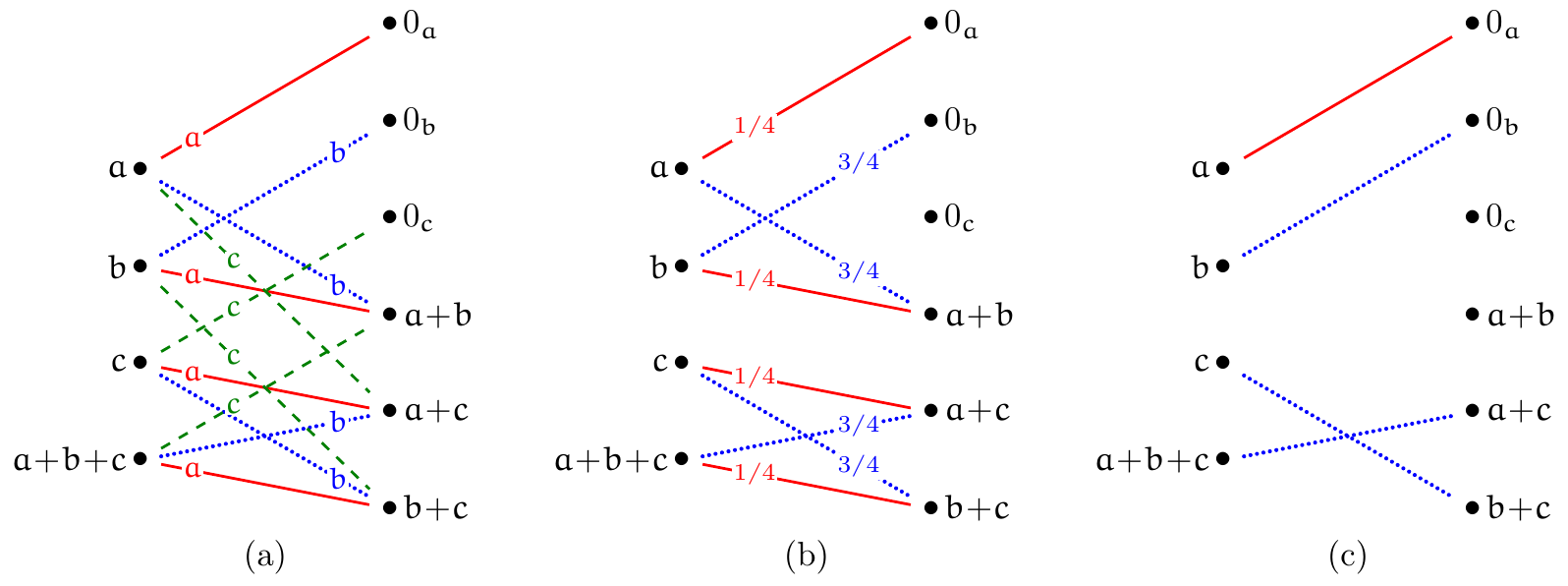}
    \caption{(a) The recovery graph of the $[7,3]$ Simplex code. (b) A fractional matching for demand vector $(1,3,0)$. (c)  An integral matching for demand vector$(1,3,0)$} 
    \label{fig:SimplexMatch}
\end{figure}
We are interested in serving the demand vector $(1,3,0)$.
An easy way to see that this vector is in the service rate region is to consider the fractional matching on the recovery graph that assigns weight $1/4$ to all $a$ recovery edges, weight $3/4$ to all $b$ recovery edges, and $0$ to all $c$ recovery edges, as shown in~\Cref{fig:SimplexMatch}-(b) and instructed by \Cref{thm:simplex-codes}. Alternatively, we can satisfy this demand with an integral service where file $a$ is downloaded solely from the node storing $a$, while file $b$ is downloaded from the node storing $b$ and the repair groups $c ~ \&  ~ b+c$ and  $a+c ~ \& ~ a+b+c$, as shown in~\Cref{fig:SimplexMatch}-(c). The matching problem asks the following: Given an $[n,k]$ code recovery graph and a matching such that the for object $i$, the sum of weights on its recovery edges is an integer $\lambda_i$, is there an integral matching with $\lambda_i$ recovery edges for object $i$, for all $1\le i\le k$?

The answer to this question is yes for the binary Simplex codes. Since these codes are batch codes \cite{wang2017switch}, the claim follows easily from the results in \Cref{sec:simplex_comb_opt}. Moreover, an algorithm is presented in~\cite{ServiceCombinatorial:KazemiKSS20} that takes an integer component demand vector $(\lambda_a,\lambda_b,\lambda_c)$ in the service rate region of the $[7,3]$ Simplex code and produces an integral matching with $\lambda_a$ $a$-recovery edges, $\lambda_b$ $b$-recovery edges, and $\lambda_c$ $c$-recovery edges (see~\cite[Algorithm 1]{ServiceCombinatorial:KazemiKSS20}). This algorithm can be easily extended to the binary $[15,4]$ Simplex code, but a generalization to an arbitrary $[2^k-1,k]$ Simplex code is an open problem.

\subsubsection{Asynchronous Service Rate Region}
\label{sec:asynchronous-service-rate-region}
In distributed systems serving multiple users, it is natural to ask the following questions. If a user leaves the system, can another user interested in downloading a different objects take the freed place? This question is of interest e.g., whenever different queries take different times to process. We refer to storage schemes that support such dynamics as \textit{asynchronous}, following \cite{Batch:RietST18} which has introduced this question and the terminology in connection with batch codes. 

Consider the $[7,3]$ Simplex code, and observe that point $(\lambda_a,\lambda_b,\lambda_c)=(1,3,0)$ belongs to its  service rate region, and can be achieved by assigning the entire $\lambda_a=1$ to the node storing $a$, and splitting
$\lambda_b=3$ evenly between the nodes storing $b$, $c ~ \&  ~ b+c$, and $a+c ~ \& ~ a+b+c$.
Suppose that all users downloading object $a$ leave the system. Can then $\lambda_c =1$ users take their place?  Although point $(\lambda_a,\lambda_b,\lambda_c)=(0,3,1)$ belongs to the service region, the answer is no because the only available servers in the system are those storing $a$ and $a+b$. Consider again the $[7,3]$ Simplex code and point $(\lambda_a,\lambda_b,\lambda_c)=(1,3,0)$, but this time the demands are split shown in~\Cref{fig:SimplexMatch}-(b), and described in the proof of \Cref{thm:simplex-codes} for the general case. Note that now the departure of all users downloading object $a$ frees the system to serve any new demand vector as long as it belongs to the service rate region.

There are two natural questions here: 1) Are there demand allocation schemes that are scalable with respect to user departures/arrivals? 2) What is the service rate region of a storage scheme if we require that each point be not only achievable but also achievable by scalable allocations? The latter question was addressed for batched codes in \cite{Batch:RietST18}, where it was found out that e.g. Simplex codes can serve any multiset of size 2 in asynchronous way, as opposed to any multiset of size 4 without this requirement. 


\section{Concluding Remarks}
In this paper we introduce the service rate region as a new aspect in the design of distributed storage and computing systems. The service rate region of a storage system storing $k$ files is the set of request demand rates $\boldsymbol{\lambda} = (\lambda_1, \lambda_2, \dots, \lambda_k)$ that can be supported by a set of $n$ servers, each of which has a limited service capacity $\mu$. Previously considered design considerations for distributed storage include reliability against node failures, repair-efficiency, data locality and latency. Codes that optimize these aspects may not be able to support a large volume of access requests, especially when different objects have different demands. The service rate region can capture this aspect and enable the design of storage schemes that maximize the volume of heterogeneous data access requests that can be satisfied with minimum number of resources.  In this paper we highlight two problems of interest: 1) optimal splitting of the requests for each object across its recovery groups in order to maximize the service rate region and 2) design of the underlying coding scheme to achieve a service rate region with desired properties. 

Through preliminary work on the first problem of optimal request splitting, we show how the notion of the service rate region employs diverse mathematical techniques such as water-filling, geometric representations of codes and combinatorial optimization over graphs. In particular, we characterize the rate regions of maximum distance separable (MDS) codes, Reed Muller codes and Simplex codes using three different techniquies: waterfilling, fractional matchings on graphs and geometric representations.

Our initial work on second thread of designing coding schemes to maximize the service rate region with a given number of servers provide the novel insight that codes that are a hybrid of replication and coding can achieve the best service rate region. Further exploration of code design for service rate region maximization can help discover fundamental connections with existing classes of codes such as batch codes and availability codes. We hope that the open problems presented in this paper will result in interdisciplinary interactions between the networking and coding theory communities and result in practical insights to boost the service capacity of distributed storage and computing systems.


\section*{Acknowledgements}
This material is in part based upon work supported by the National Science Foundation under Grant No.~CIF-1717314. We thank Sarah Anderson, Ann Johnston, Gretchen Matthews, Esmaeil Karimi, Carolyn Mayer, and Gala Yadgar for helpful discussions.

\newpage
\bibliographystyle{IEEEtran}
\bibliography{Bib_coding_schemes, storage_thesis}

\end{document}